\documentclass[journal,10pt, twocolumn, twoside]{IEEEtran}
\IEEEoverridecommandlockouts

\usepackage[table]{xcolor} 
\renewcommand{\arraystretch}{2.5}
\newcolumntype{C}[1]{>{\centering\arraybackslash}p{#1}}
\usepackage{graphicx}
\usepackage{amsmath,amsxtra, amsfonts, amssymb}
\usepackage{setspace}
\usepackage{times}
\usepackage{epsfig}
\usepackage{subcaption}
\usepackage{latexsym}
\usepackage{cite}
\usepackage{url}
\usepackage{epstopdf}
\usepackage{tikz,pgf,pgfplots}
\usepackage{multirow}
\usepackage{cool} 
\usepackage{flushend}
\usepackage{mathtools}
\usepackage{dsfont}
\usepackage{amsthm}
\usepackage[acronym,shortcuts]{glossaries}

\usepackage{makecell}
\usepackage{array}
\newcolumntype{?}{!{\vrule width 0.6\arrayrulewidth}}

\newtheorem{theorem}{Theorem}
\newtheorem{lemma}{Lemma}
\newtheorem{proposition}{Proposition}
\newtheorem{remark}{Remark}
\newtheorem{corollary}{Corollary}
\newtheorem{definition}{Definition}

\newcommand{\Psr}{\mathcal{P}_{\rm{sr}}}
\newcommand{\PsrR}{\mathcal{P}_{\rm{sr,RF}}}

\newcommand{\Prd}{\mathcal{P}_{\rm{rd}}}
\newcommand{\PrdR}{\mathcal{P}_{\rm{rd,RF}}}
\newcommand{\Psrbar}{\overline{\mathcal{P}}_{\rm{sr}}}
\newcommand{\Prdbar}{\overline{\mathcal{P}}_{\rm{rd}}}
\newcommand{\Pout}{\mathcal{P}_{\rm{E-E}}}
\newcommand{\PoutR}{\mathcal{P}_{\rm{E-E,RF}}}

\newcommand{\PUB}{f}
\newcommand{\PR}{x}

\newcommand{\skipval}{0.87mm} 
\newacronym{5G}{5G}{fifth generation}
\newacronym{AWGN}{AWGN}{additive white Gaussian noise}
\newacronym{CDF}{CDF}{cumulative distribution function}
\newacronym{CSI}{CSI}{channel state information}
\newacronym{FDR}{FDR}{full-duplex relaying}
\newacronym{HDR}{HDR}{half-duplex relaying}
\newacronym{IC}{IC}{interference channel}
\newacronym{IGS}{IGS}{improper Gaussian signaling}
\newacronym{MHDF}{MHDF}{multi-hop decode-and-forward}
\newacronym{MIMO}{MIMO}{multiple-input multiple-output}
\newacronym{MISO}{MISO}{multiple-input single-output}
\newacronym{MRC}{MRC}{maximum ratio combining}
\newacronym{PDF}{PDF}{probability density function}
\newacronym{PGS}{PGS}{proper Gaussian signaling}
\newacronym{RSI}{RSI}{residual self-interference}
\newacronym{RV}{RV}{random variable}
\newacronym{SISO}{SISO}{single-input single-output}

\newacronym{CEMSE}{CEMSE}{Computer, Electrical, and Mathematical Sciences and Engineering}
\newacronym{KAUST}{KAUST}{King Abdullah University of Science and Technology}

\begin{document}
\title{Full-Duplex Relaying with Improper Gaussian Signaling over Nakagami-$m$ Fading Channels}
\author{\IEEEauthorblockN{Mohamed Gaafar, \textit{Student Member, IEEE}, Mohammad Galal Khafagy, \textit{Senior Member, IEEE}, Osama Amin, \\\textit{Senior Member, IEEE}, Rafael F. Schaefer, \textit{Senior Member, IEEE} and Mohamed-Slim Alouini, \textit{Fellow, IEEE}
\thanks{M. Gaafar and R. F. Schaefer are with the Information Theory and Applications Chair, Technische Universit\"at Berlin, Germany. Email:  \{mohamed.gaafar, rafael.schaefer\}@tu-berlin.de.
\newline \indent O. Amin and M.-S. Alouini are with the \ac{CEMSE} Divison, \ac{KAUST}, Thuwal, Makkah Province, Saudi Arabia. Email: \{osama.amin, slim.alouini\}@kaust.edu.sa.
\newline \indent M. G. Khafagy was with the \ac{CEMSE} Division, \ac{KAUST}, Thuwal, Makkah Province, Saudi Arabia. Currently, he is with the Computer Science and Engineering Department, College of Engineering, Qatar University, Doha, Qatar. Email: mohammad.khafagy@kaust.edu.sa. }}
}


\maketitle
\begin{abstract}
We study the potential employment of \ac{IGS} in \ac{FDR} with non-negligible \ac{RSI} under Nakagami-$m$ fading. \Ac{IGS} is recently shown to outperform traditional \ac{PGS} in several interference-limited settings. In this work, \ac{IGS} is employed as an attempt to alleviate \ac{RSI}. We use two performance metrics, namely, the outage probability and the ergodic rate. First, we provide upper and lower bounds for the system performance in terms of the relay transmit power and circularity coefficient, a measure of the signal impropriety. Then, we numerically optimize the relay signal parameters based only on the channel statistics to improve the system performance. Based on the analysis, \ac{IGS} allows \ac{FDR}  to operate even with high \ac{RSI}.  The results show that \ac{IGS} can leverage higher power budgets to enhance the performance, meanwhile it relieves  \ac{RSI} impact via tuning the signal impropriety. Interestingly, one-dimensional optimization of the circularity coefficient, with maximum relay power, offers a similar performance as the joint optimization, which reduces the optimization complexity. From a throughput standpoint, it is shown that \ac{IGS}-\ac{FDR} can outperform not only \ac{PGS}-\ac{FDR}, but also \acl{HDR} with/without \acl{MRC} over certain regions of the target source rate.

\end{abstract}

\glsresetall

\section{Introduction}
Contrary to a long-held acceptance that radio front-ends cannot simultaneously transmit and receive, a truly promising potential for \emph{full-duplex communications} has been shown by recent hardware developments \cite{mobicom2011fullduplex},\cite{201212_TWC_Duarte}. Indeed, by multiplexing inbound and outbound traffic over the same channel resource, a full-duplex radio can recover the spectral efficiency loss known to be encountered by its half-duplex counterpart. Performance merits of full-duplex radio have been recently investigated in different communication settings, including full-duplex bidirectional communication, full-duplex base stations, and \ac{FDR} \cite{201409_JSAC_FD_Tutorial}, with the latter being the focus of this work.  These merits have qualified full-duplex communication to be considered as a candidate technology for future \ac{5G} wireless networks \cite{201402_COMMAG_5G_full_duplex}. 

\Ac{FDR} allows a relay node to listen to an information source and simultaneously forward to its intended destination. Theoretically, this simultaneous transmission/reception doubles the spectral efficiency in the relay channel. However, in practice, this comes at the expense of a self-interference level introduced at the receiver of the relay node from its own transmitter. Even with the application of advanced self-interference isolation and cancellation techniques, a level of \ac{RSI} persists. Such a persistent \ac{RSI} link and the means to mitigate it represent the main challenge in full-duplex communications, especially with the fact that its adverse effect can typically be an increasing function of the relay power. Therefore, increasing the relay power no longer guarantees an enhanced end-to-end performance. For instance, by increasing the relay power in a fixed-rate transmission scheme, the relay may forward more reliably to the destination in the second hop. However, it also increases the \ac{RSI} level which negatively affects the reliability in the first hop. Hence, higher relay power budgets cannot be always utilized beyond a certain threshold. Consequently, employing interference mitigation strategies in \ac{FDR} is crucial to attain a satisfactory performance of full-duplex transmissions.

\Ac{IGS} has been first introduced in \cite{cadambe2010interference}, where higher degrees of freedom for the $3$-user  \ac{SISO} \ac{IC} were shown to be achievable. This comes in contrary to other communication settings with interference-free channels where \ac{PGS} is the optimal choice. \Ac{PGS} assumes the zero-mean complex Gaussian transmit signal to be statistically circularly symmetric with uncorrelated real and imaginary components. On the other hand, \ac{IGS} is a class of signals where circularity and uncorrelatedness conditions can be relaxed \cite{schreier2010statistical}. The results in \cite{cadambe2010interference} motivate the need to further study the potential gains of \ac{IGS} in communication scenarios where interference imposes a noticeable limitation. 

\Ac{IGS} has been recently adopted to improve the performance of different interference-limited communication settings, namely, \ac{SISO}-\ac{IC} \cite{ho2012improper,lameiro2013degrees}, \ac{MISO}-\ac{IC} \cite{zeng2013optimized}, \ac{MIMO}-\ac{IC} \cite{wang2011optimality,lagen2016coexisting}, Z-\ac{IC} \cite{kurniawan2015improper, lagen2016improper, lameiro2016rate}, \ac{MIMO} X-channel \cite{yang2014interference}, interference broadcast channels \cite{shin2012new}, cognitive radio channels using underlay \cite{lameiro2015benefits,Gaafar2015Spectrum,lameiro2016maximally,amin2016underlay,gaafar2017underlay} and overlay \cite{amin2017overlay} communication paradigms, alternate (two-path, virtual full-duplex) relaying \cite{gaafar2016letter}, symbol error rate reduction \cite{nguyen2015improper} and asymmetric hardware distortions ~\cite{sidrah2017asymmetric}.

The potential gains of \ac{IGS} have been also recently studied in \cite{Hellings2014OnOptimal} for the \ac{MIMO} relay channel when a partial decode-and-forward strategy is adopted. In such a relaying strategy, the relay only decodes a part of the message, while the rest of the message is treated as an additional interference term. It was shown in \cite{Hellings2014OnOptimal} that \ac{PGS} is optimal within the class of Gaussian signals. However, the work in \cite{Hellings2014OnOptimal} assumed an ideal \ac{FDR}, where the self-interference imposed by the relay's transmitter on its own receiver is perfectly canceled. In \cite{kim2012asymmetric}, a simplified model for \ac{FDR} with \ac{IGS} was considered, where the transmit signal only occupies the real dimension of the two-dimensional complex signal space. Also, all the communicating nodes in \cite{kim2012asymmetric} are assumed to perfectly share the instantaneous Rayleigh-fading channel coefficients. For such a simplified model, it was shown that \ac{IGS} was able to effectively improve the achievable rate by eliminating the \ac{RSI} via its alignment in only one orthogonal signal space dimension, and decoding the desired signal from the other. The more general problem, however, remains more challenging with further interesting aspects to investigate, where complex transmit signals are employed and only channel statistics are made available at the transmitter side.

In this work\footnote{While this work was in progress, preliminary results have been accepted and presented in IEEE ICC'16 \cite{gaafar2016improper}. In this work, we consider a more generalized framework under Nakagami-$m$ fading, which includes the work in \cite{gaafar2016improper} as a special case. We present herein the derived expressions for another performance metric, the ergodic rate, to evaluate the performance of \ac{IGS} in \ac{FDR}. Furthermore, we provide more detailed analysis and insights addressing the outage performance. Finally, more numerical results are presented to verify the mathematical analysis.}, we study the potential performance merits of employing \ac{IGS} in \ac{FDR} systems with non-negligible \ac{RSI} over Nakagami-$m$ fading channels. The main contributions of this paper can be summarized as follows: 
\begin{itemize}
\item We first derive the exact end-to-end outage probability of the \ac{FDR} system adopting \ac{IGS} at the relay over Nakagami-$m$ fading as a function of the relay's transmit power and circularity coefficient  in an  integral form. Further, we derive a lower bound of the end-to-end  outage probability in closed form. Moreover, we provide simpler forms of the derived expressions for Rayleigh fading and derive an upper bound for the end-to-end outage probability based on some convexity properties that are also proven herein. Based on this upper bound, we present an asymptotic analysis that shows the benefits which can be reaped from \ac{IGS} in \ac{FDR} systems.
\item We derive an integral form of the exact end-to-end ergodic rate of the \ac{IGS}-\ac{FDR} system as a function of the relay's transmit power and circularity coefficient over Nakagami-$m$ fading. In addition, we derive an upper bound for the end-to-end  ergodic rate. We also present a lower bound for the end-to-end ergodic rate over Rayleigh fading.
\item We prove unimodality properties of the presented end-to-end outage probability upper bound in the relay's transmit power and circularity coefficient over Rayleigh fading, which allow for efficient numerical optimization using standard tools. Also, we design the relay's transmit signal characteristics,  represented in the power and circularity coefficient, by a coordinate descent algorithm based on the outage upper bound over Rayleigh fading.
\item Finally, with the aid of numerical results, we first validate the derived bounds for the outage probability and ergodic rate. Next, we discuss the benefits that can be reaped by employing \ac{IGS} in \ac{FDR} in comparison to \ac{PGS}, and also relative to \ac{HDR}, under different system and fading parameters. We show that \ac{IGS}-\ac{FDR}  is not only able to outperform \ac{PGS}-\ac{FDR} in terms of the end-to-end throughput, but also that of \ac{HDR} even with \ac{MRC} over certain ranges of the targeted source rate. Finally, we show that one may not seriously compromise the benefits attained via the joint tuning of the power and circularity, and still attain a close end-to-end performance merits through a simple one-dimensional tuning of the circularity coefficient.   
\end{itemize}
The rest of the paper is organized as follows. In Section II, we describe the adopted \ac{FDR}  system model. Section III derives the outage probability of the \ac{FDR} system over Nakagami-$m$ fading, while Section IV deals with the ergodic rate of the system. The design of the signal characteristics for the  relay transmissions based on the outage probability over Rayleigh fading is presented in Section V. We validate the performance of \ac{IGS} in \ac{FDR} systems in Section VI through numerical simulations. Finally, we conclude the paper in Section VII. \\

\textbf{\textit{Notation and Special Functions:}} Throughout the rest of the paper, we use $ \left| \cdot \right| $ to denote the absolute value operation. $\mathbb{P}\{A\}$ denotes the probability of occurrence of an event $A$. The operator $\mathbb{E}\{\cdot\}$ is used to denote the statistical expectation, with the mean of a \ac{RV} $X$ defined as $\bar{X} = \mathbb{E}\{ X \}$. $\binom nk$ is the binomial coefficient and $\min \left(x,y\right)$ is the minimum of the two quantities $x$ and $y$. Throughout the analysis, we use the following special functions \cite{abramowitz1966handbook}.  $\Gamma(a)=\int\nolimits_0^\infty  {{{t^{a-1}}}}{{{{e^{ - t}}}}} dt$ denotes the gamma function and $\Gamma(a,x)=\int\nolimits_x^\infty  {{{t^{a-1}}}}{{{{e^{ - t}}}}} dt$ is the upper incomplete gamma function. ${E_n}(x) = \int\nolimits_1^\infty  {\frac{{{e^{ - xt}}}}{{{t^n}}}} dt$ and $U(a,b,z) = \frac{1}{{\Gamma \left( a \right)}}\int\nolimits_0^\infty  {{t^{a - 1}}{{\left( {t + 1} \right)}^{b - a - 1}}{e^{ - zt}}dt}\;, a, z > 0 $ are the exponential integral and the Tricomi's confluent hypergeometric functions, respectively.  When $n=1$, we use $E_1(x)=-\rm{Ei}(-x),\;x>0$, where the exponential integral $\rm{Ei}(x)=-\int\nolimits_{-x}^\infty  {\frac{{{e^{ - t}}}}{{{t}}}} dt$.  Also, we use ${\Xi _n}\left( x \right) = {e^x}{E_n}(x)$ to combine the exponential function and integral.

\section{System Model}\label{sec:sys_mod}
\begin{figure}[!t]
\centering
\includegraphics[width=0.7\columnwidth]{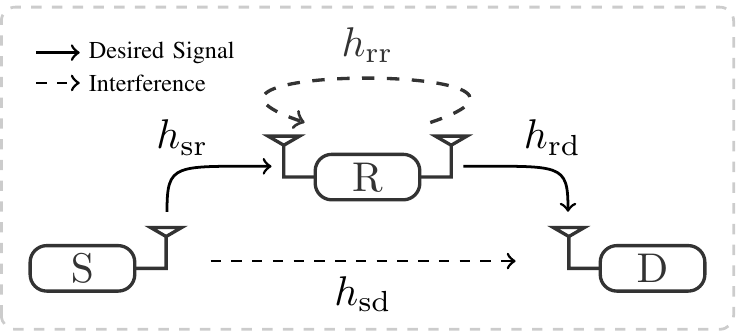}
\caption{A full-duplex cooperative setting in coverage extension scenarios.}
\label{sysmodfig}
\end{figure}
We consider the communication setting depicted in Fig. \ref{sysmodfig}, where a source (${\rm S}$) intends to communicate with a distant destination (${\rm D}$). The direct ${\rm S -  D}$ link is assumed of a relatively weak gain due to path loss and shadowing effects. Accordingly, a full-duplex relay (${\rm R}$) is utilized to assist the end-to-end communication and extend the coverage.  \Ac{FDR} can offer higher spectral efficiency when compared to \ac{HDR}. However, \ac{FDR} in practice suffers from an \ac{RSI} level which imposes an additional communication challenge. In addition, the received signal component via the ${\rm S -  D}$ link is assumed to be weak and hence, it is considered as interference at the destination for simpler decoding purposes as commonly assumed in the literature \cite{kwon2010optimal,Riihonen2011Hybrid}. Thus, the \ac{FDR} system under consideration suffers from two interference sources; the \ac{RSI} at the relay, and the direct ${\rm S -  D}$ link signal received at the destination. 
\subsection{Channel Model}\label{subsec:chan_model}
The fading coefficient of the $i-j$ link is denoted by $h_{ij}$, for $i\in\{\rm{s},\rm{r}\}$ and $j\in\{\rm{r},\rm{d}\}$, where $\rm{s}$, $\rm{r}$ and $\rm{d}$ refer to the source, relay, and destination nodes, respectively. Moreover, the $i-j$ link gain is denoted by $g_{ij}=|h_{ij}|^2$. All channels are assumed to follow a block fading model, where $h_{ij}$ remains constant over one block, and varies independently from one block to another following a 
Nakagami-$m$ fading distribution with a shaping parameter $m_{ij}$ and average power $\mathbb{E}\left\{|h_{ij}|^2\right\}=\pi_{ij}$. Accordingly, the channel gain $g_{ij}$ is a gamma \ac{RV} with a shaping parameter\footnote{In this paper, we assume only \textit{integer} shape parameter.  Typical analysis involving wireless communications over Nakagami-$m$ channels uses values of $m=1$ up to $m=4$ to keep the fading effects \cite{simon2005digital}.  } $m_{ij}$ and a scale parameter $\theta_{ij}=\frac{\pi_{ij}}{m_{ij}}$. Furthermore, it should be noted that the \ac{PDF} of $g_{ij}$ is defined as
\begin{equation}
f_{g_{ij}}(x;m_{ij},\theta_{ij})=\frac{x^{m_{ij}-1}e^{-\frac{x}{\theta_{ij}}}}{\Gamma(m_{ij})\theta_{ij}^{m_{ij}}},\;\;x\geq 0\cdot
\end{equation}
In the special case $m_{ij}= 1$, Rayleigh fading is obtained. All channel fading gains are assumed to be mutually independent. 

The relay operates in a full-duplex mode where simultaneous listening/forwarding is allowed with an introduced level of loopback interference. The link coefficient $h_{\rm{rr}}$\footnote{Depending on the communication setting, transmit power ranges, and the employed RSI cancellation techniques, the average RSI link gain is typically assumed to range from 3 to 10/15 dB above the noise floor \cite[Section IV.B]{201212_TWC_Duarte},  \cite[Chapter 2]{Khafagy2016FDRDissertation},  \cite{2013_SIGCOMM_Stanford_Single_FD}.} is assumed to represent the \ac{RSI} after undergoing all possible isolation and cancellation techniques \cite{kwon2010optimal,Riihonen2011Hybrid, khafagy2015efficient}. The source and the relay powers are denoted by $P_{\rm{s}}$ and $P_{\rm{r}}$, respectively, where both are restricted to a maximum allowable value of $P_{\mathrm{max}}$. Also, $n_{\rm{r}}$ and $n_{\rm{d}}$ denote the circularly-symmetric complex \ac{AWGN} components at the relay and the destination, with variance $\sigma_{\rm{r}}^2$ and $\sigma_{\rm{d}}^2$, respectively. Without loss of generality, we assume that $\sigma_{\rm{r}}^2=\sigma_{\rm{d}}^2=1$. 
\subsection{Signal Model}
The transmit signals at the source and the relay at time $t$ are denoted by $x_{\rm{s}}[t]$ and $x_{\rm{r}}[t]$, respectively.  \ac{PGS} is adopted at the source \footnote{For the ease of exposition, we use PGS at the source. This assumption can be further justified since the use of IGS at the source calls for a joint source/relay improper signal optimization. Although an IGS source is expected to offer further gains, such a joint optimization turns out to be mathematically involved and intractable.} which transmits with its  power budget, $P_{\mathrm{s}}=P_{\mathrm{max}}$.  On the other hand, with the availability of statistical transmit \ac{CSI} at the relay, zero-mean \ac{IGS} is adopted at the relay in order to mitigate the non-negligible \ac{RSI} at its receiver. The degree of impropriety of $x_{\mathrm{r}}[t]$ is measured based on the following definitions. 
\begin{definition}\label{def:1}\cite{Neeser1993proper}
The variance and pseudo-variance of the relay's transmit signal, $x_{\mathrm{r}}[t]$, are given by $\sigma _x^2={\mathbb{E}} {{{\{\left| x_{\rm{r}} \right|}^2\}}}$ and $\tilde\sigma _x^2=\mathbb{E} {{\{{ x_{\rm{r}}^2 }\}}}$, respectively. 
\end{definition}
\begin{definition}\label{def:2}
A signal is called \emph{proper} if it has a zero pseudo-variance $\tilde\sigma _x^2$, while an \emph{improper} signal has a non-zero $\tilde\sigma _x^2$. 
\end{definition}
\begin{definition}\label{def:3}\cite{lameiro2015benefits}
A circularity coefficient is a measure of the degree of impropriety of the signal $x_{\mathrm{r}}[t]$, which is given as ${{\cal C}_x} = {{\left| {\tilde \sigma _x^2} \right|}}/{{\sigma _x^2}}$.
\end{definition}
\noindent \noindent Following from Definitions \ref{def:2} and \ref{def:3}, the circularity coefficient satisfies $0\leq\mathcal{C}_x \leq 1$. In particular, $\mathcal{C}_x=0$ implies a \emph{proper} signal, while $\mathcal{C}_x=1$ implies a \emph{maximally improper} signal.

 The received signals at the relay and the destination at time $t$ are given, respectively, by
\begin{eqnarray}
y_{\rm{r}}[t]	
&=&
\sqrt{P_{\rm{s}}} h_{\rm{sr}} x_{\rm{s}}[t]+ \sqrt{P_{\rm{r}}} h_{\rm{rr}}x_{\rm{r}}[t] + n_{\rm{r}}[t],\label{relay_rec_signal}\\
y_{\rm{d}}[t]	
&=&
\sqrt{P_{\rm{r}}} h_{\rm{rd}} x_{\rm{r}}[t]+ \sqrt{P_{\rm{s}}} h_{\rm{sd}}x_{\rm{s}}[t] + n_{\rm{d}}[t]\cdot\label{dest_rec_signal}
\end{eqnarray}
The relay is assumed to adopt a decode-and-forward relaying strategy, where it does not transmit any message of its own, but forwards the regenerated source message after decoding. Due to the source and relay asynchronous transmissions, the signal transmitted by the relay (source) is considered as an additional noise term at the relay (destination) in the decoding stage as commonly treated in the related literature \cite{kwon2010optimal,Riihonen2011Hybrid}. 
\subsection{Achievable Rates with \ac{IGS}}
From the adopted signal model in \eqref{relay_rec_signal} and \eqref{dest_rec_signal}, each transmit signal, i.e., from the source and the relay transmitter, is considered as a desired signal at one receiver while treated as interference at the other. Hence, the rate expressions for the first and second hops have the same form of those of a two-user IC.
\begin{lemma}\label{single_link_rate}\cite{zeng2013transmit}
The achievable rate of a single link that is subjected to  interference-plus-noise  $z$ when $x$ is transmitted and observed as $y$ is expressed as 
\begin{equation} 
R = \frac{1}{2}\log_2 \left( {\frac{{\sigma _{{{{y}}}}^4 - {{\left| {\tilde \sigma _{{{{y}}}}^2} \right|}^2}}}{{\sigma _z^4 - {{\left| {\tilde \sigma _z^2} \right|}^2}}}} \right)\cdot
\end{equation}
\end{lemma} 
 As a result of using \ac{IGS}, while treating the interference as a Gaussian noise, the achievable rates of the \ac{FDR} system are analyzed based on Lemma \ref{single_link_rate}. First,  the achievable rate supported by the ${\rm S -  R}$ link can be expressed as
\begin{align}\label{R_sr_zeng}
\!\!\!{{R}_{\rm{sr}}}\!\left( { {P_{\rm{r}}},{{\cal C}_x}} \right) &={\log _2}\left(\!{1 + \frac{{{P_{\rm{s}}}{g_{\rm{sr}}}}}{{{P_{\rm{r}}}{g_{\rm{rr}}} + 1}}}\!\right)\!+\!\frac{1}{2}{\log _2}\left( {\frac{{1 - {\cal C}_{{{{y}}_{\rm{r}}}}^2}}{{1 - {\cal C}_{{{{I}}_{\rm{r}}}}^2}}} \right),\!\!
\end{align} 
where ${\cal C}_{{{{y}}_{\rm{r}}}}$ and ${\cal C}_{{{{I}}_{\rm{r}}}}$ are the circularity coefficients of the received and
interference-plus-noise signals at the relay, respectively.
Hence, \eqref{R_sr_zeng} can be simplified as
\begin{align}\label{R_sr_zeng_simplified}
\!\!\!\!\!\!{{R}_{\rm{sr}}}\!\left( { {P_{\rm{r}}},{{\cal C}_x}} \right) \!&=\!\frac{1}{2} {\log _2}\!\left(\!{\frac{{{{\left( {{P_{\rm{s}}}{g_{\rm{sr}}}+ {P_{\rm{r}}}{g_{\rm{\rm{rr}}}} + 1} \right)}^2}\!-\!{{\left( {{P_{\rm{r}}}{g_{\rm{rr}}}{{\cal C}_x}} \right)}^2}}}{{{{\left( {{P_{\rm{r}}}{g_{\rm{rr}}} + 1} \right)}^2}\!-\!{{\left( {{P_{\rm{r}}}{g_{\rm{rr}}}{{\cal C}_x}} \right)}^2}}}}\!\right)\!\cdot\!\!\!
\end{align}
Similarly, the achievable rate supported by the ${\rm R -  D}$ link is given by
\begin{align}\label{R_rd_zeng_simplified}
\!\!\!\!\!\!{{R}_{\rm{rd}}}\!\left( { {P_{\rm{r}}},{{\cal C}_x}} \right) \!&=\!\frac{1}{2}{\log _2}\!\left(\!{\frac{{{{\left( {{P_{\rm{r}}}{g_{\rm{rd}}}\!+\!{P_{\rm{s}}}{g_{\rm{sd}}} + 1} \right)}^2}\!-\!{{\left( {{P_{\rm{r}}}{g_{\rm{rd}}}{{\cal C}_x}} \right)}^2}}}{{{{\left( {{P_{\rm{s}}}{g_{\rm{sd}}} + 1} \right)}^2}}}}\!\right)\!\cdot\!\!\!
\end{align}    
One can notice that if $\mathcal{C}_x=0$, we obtain the well known achievable rates of \ac{PGS}. From \eqref{R_sr_zeng_simplified} and  \eqref{R_rd_zeng_simplified}, by adopting \ac{IGS} at the relay transmitter, the rate of the ${\rm S -  R}$ hop improves while the  rate of the ${\rm R -  D}$ hop deteriorates which creates a trade-off that can be optimized to yield a better performance than \ac{PGS}.
\section{Outage Performance}
In this section, we derive the end-to-end outage probability of the \ac{FDR} system  depicted in Fig. \ref{sysmodfig} under the assumption of Nakagami-$m$ fading while the relay transmits improper signals. Moreover, we obtain simpler forms for the Rayleigh fading case.
%
%
 The end-to-end outage probability is given by
\begin{eqnarray}\label{p_out_overall}
\Pout &=& 1- \Psrbar~\Prdbar,
\end{eqnarray}
where $\Psr$ and $\Prd$ denote the outage probability in the ${\rm S - R}$ and the ${\rm R - D}$ links, respectively, while $\overline{\mathcal{P}}_{ij}=1-\mathcal{P}_{ij}$. In what follows, we derive the outage probability expressions in the individual links, i.e., $\Psr$ and $\Prd$.
\subsection{Outage Probability of ${\rm S - R}$ Link}
 Let $r$ (bits/sec/Hz) denote the target rate of the ${\rm S - R}$ link, then its outage probability is defined ~as
 \begin{equation}\label{p_out_sr_prob_m}
\Psr\left( { {P_{\rm{r}}},{{\cal C}_x}} \right) = \mathbb{P}\left\{ {{{R}_{\rm{sr}}}\left( { {P_{\rm{r}}},{{\cal C}_x}} \right) < r} \right\}\cdot
 \end{equation}
The outage probability of the ${\rm S - R}$ can be calculated from Lemma \ref{OP_S-R_int}.
 \begin{lemma}\label{OP_S-R_int}
 The exact outage probability of the ${\rm S - R}$ link with a target rate $r$ is given by
 \begin{align}
&{{\cal P}_{\rm{sr}}}\left( {{P_{\rm{r}}},{{\cal C}_x}} \right) = 1 - \frac{1}{{\Gamma \left( {{m_{\rm{rr}}}} \right)\Gamma \left( {{m_{{\rm{sr}}}}} \right)\theta _{\rm{rr}}^{{m_{\rm{rr}}}}}} \nonumber \\
&\times \int\limits_0^\infty  {{x^{{m_{\rm{rr}}} - 1}}\Gamma \left( {{m_{{\rm{sr}}}},\frac{{\left( {{P_{\rm{r}}}x + 1} \right)}}{{{P_{\rm{s}}}{\theta _{{\rm{sr}}}}}}\Psi_r \left( {\frac{{{P_{\rm{r}}}x{{\cal C}_x}}}{{{P_{\rm{r}}}x + 1}}} \right)} \right){e^{ - \frac{x}{{{\theta _{\rm{rr}}}}}}}} dx, \label{lemma_1_Psr}
\end{align}
where $\Psi_r\left( x\right) = \left( {\sqrt {1 + \gamma \left( {1 - {x^2}} \right)}  - 1} \right)$ and  $\gamma  = {2^{2r}} - 1$.
 \end{lemma}
 
 \begin{proof}
 By substituting \eqref{R_sr_zeng_simplified} into \eqref{p_out_sr_prob_m}, we get
 \begin{align}
{{\cal P}_{\rm{sr}}}\left( {{P_{\rm{r}}},{{\cal C}_x}} \right) = \mathbb{P}\Big\{ &P_{\rm{s}}^2g_{\mathrm{sr}}^2 + 2{P_{\rm{s}}}{g_{\mathrm{sr}}}\left( {{P_{\rm{r}}}{g_{\mathrm{rr}}} + 1} \right) \nonumber\\
& - \gamma \left( {{{\left( {{P_{\rm{r}}}{g_{\mathrm{rr}}} + 1} \right)}^2} - {{\left( {{P_{\rm{r}}}{g_{\mathrm{rr}}}{{\cal C}_x}} \right)}^2}} \right) < 0 \Big\}\cdot
\end{align}
The expression inside the probability expression is a quadratic function in $g_{\mathrm{sr}}$. Hence, the conditional outage probability given  $g_{\mathrm{rr}}$ is equivalent to integrating the gamma \ac{PDF} of  $g_{\mathrm{sr}}$ over the region in which the quadratic function is less than $0$ which gives
\begin{align}\label{OP_S-R_conditioned}
{{\cal P}_{\rm{sr}}}\left( {{P_{\rm{r}}},{{\cal C}_x}\left| {{g_{\mathrm{rr}}}} \right.} \right) &= \frac{1}{{\Gamma \left( {{m_{{\rm{sr}}}}} \right)\theta _{{\rm{sr}}}^{{m_{{\rm{sr}}}}}}}\int\limits_0^{g_{\mathrm{sr}}^\circ } {{x^{{m_{{\rm{sr}}}} - 1}}{e^{ - \frac{x}{{{\theta _{{\rm{sr}}}}}}}}\;} dx  \nonumber \\
&= 1 - \frac{{\Gamma \left( {{m_{{\rm{sr}}}},\frac{{g_{\mathrm{sr}}^\circ }}{{{\theta _{{\rm{sr}}}}}}} \right)}}{{\Gamma \left( {{m_{{\rm{sr}}}}} \right)}},
\end{align} 
where $g_{\mathrm{sr}}^\circ  = \frac{{\left( {{P_{\rm{r}}}{g_{\mathrm{rr}}} + 1} \right)}}{{{P_{\rm{s}}}}}{\Psi_r }\left( {\frac{{{P_{\rm{r}}}{g_{\mathrm{rr}}}{{\cal C}_x}}}{{{P_{\rm{r}}}{g_{\mathrm{rr}}} + 1}}} \right)$ represents the positive root of the inequality inside the probability sign. Therefore, by averaging over the statistics of $g_{\rm{rr}}$ in \eqref{OP_S-R_conditioned}, we obtain
\begin{align}
{{\cal P}_{\rm{sr}}}\left( {{P_{\rm{r}}},{{\cal C}_x}} \right) = 1-\mathbb{E}_{g_{\rm{rr}}}\left\{\frac{{\Gamma \left( {m_{{\rm{sr}}}},{{\frac{{\left( {{P_{\rm{r}}}{g_{\mathrm{rr}}} + 1} \right)}}{{{P_{\rm{s}}}{{{\theta _{{\rm{sr}}}}}}}}{\Psi_r }\left( {\frac{{{P_{\rm{r}}}{g_{\mathrm{rr}}}{{\cal C}_x}}}{{{P_{\rm{r}}}{g_{\mathrm{rr}}} + 1}}} \right)}} \right)}}{{\Gamma \left( {{m_{{\rm{sr}}}}} \right)}} \right\}\!\!,
\end{align}
which directly yields the result in \eqref{lemma_1_Psr}.
 \end{proof}
Unfortunately, there is no closed form  expression for the integral in Lemma \ref{OP_S-R_int}. In the following, we derive a lower bound on the ${\rm S - R}$ outage probability for Nakagami-$m$ fading, in addition to an upper bound for the special Rayleigh fading case.
\vspace{5pt}
\subsubsection{Lower Bound}
A lower bound on ${{\cal P}_{\rm{sr}}}$ is provided in the following lemma.
\begin{lemma}\label{OP_S-R_LB}
The ${\rm S -  R}$ link outage probability can be lower-bounded as
\begin{align}
{{\cal P}_{{\rm{sr}}}} &\ge 1 - \frac{{{e^{ - \frac{{{\Psi _{{r}}}\left( {{{\cal C}_x}} \right)}}{{{P_{\rm{s}}}{\theta _{{\rm{sr}}}}}}}}}}{{\Gamma \left( {{m_{{\rm{rr}}}}} \right)\theta _{{\rm{rr}}}^{{m_{{\rm{rr}}}}}}}  \nonumber\\
&\times\sum\limits_{m = 0}^{{m_{{\rm{sr}}}} - 1} {\sum\limits_{k = 0}^m \binom mk \frac{{P_{\rm{r}}^k\Gamma \left( {k + {m_{{\rm{rr}}}}} \right){{\left( {\frac{{{\Psi _{{r}}}\left( {{{\cal C}_x}} \right)}}{{{P_{\rm{s}}}{\theta _{{\rm{sr}}}}}}} \right)}^m}}}{{\Gamma \left( {m + 1} \right){{\left( {\frac{{{P_{\rm{r}}}{\Psi _{{r}}}\left( {{{\cal C}_x}} \right)}}{{{P_{\rm{s}}}{\theta _{{\rm{sr}}}}}} + \frac{1}{{{\theta _{{\rm{rr}}}}}}} \right)}^{k + {m_{{\rm{rr}}}}}}}}} \nonumber \\
&  \buildrel \Delta \over = {\cal P}_{{\rm{sr}}}^{{\rm{LB}}}\left( {{P_{\rm{r}}},{{\cal C}_x}} \right)\cdot
\end{align}
\end{lemma}
\begin{proof}
The ${\rm S - R}$  outage probability expression in Lemma \ref{OP_S-R_int} can be lower-bounded by replacing ${\frac{{{P_{\rm{r}}}x}}{{{P_{\rm{r}}}x + 1}}}{{\cal C}_x}$ by $\mathcal{C}_x$. This can be easily seen as $\Psi_r\left( x\right)$ and $\Gamma\left( .,x\right)$ are monotonically decreasing in $x$. Then, we use the series expansion of the upper incomplete gamma function \cite[8.352-2]{gradstein1980tables} as\\[0.1cm]
\begin{align}
&\Gamma \left( {{m_{{\rm{sr}}}},\frac{{\left( {{P_{\rm{r}}}x + 1} \right){\Psi _r}\left( {{{\cal C}_x}} \right)}}{{{P_{\rm{s}}}{\theta _{{\rm{sr}}}}}}} \right) \nonumber \\
&=  \Gamma \left( {{m_{{\rm{sr}}}}} \right){e^{ - \frac{{\left( {{P_{\rm{r}}}x + 1} \right){\Psi _r}\left( {{{\cal C}_x}} \right)}}{{{P_{\rm{s}}}{\theta _{{\rm{sr}}}}}}}} \sum\limits_{m = 0}^{{m_{{\rm{sr}}}} - 1} {\frac{{{{\left( {\frac{{\left( {{P_{\rm{r}}}x + 1} \right){\Psi _r}\left( {{{\cal C}_x}} \right)}}{{{P_{\rm{s}}}{\theta _{{\rm{sr}}}}}}} \right)}^m}}}{{\Gamma \left( {m + 1} \right)}}}\cdot
\end{align} 
Hence, the integral can be rewritten as
\begin{align}\label{out_S-R_temp}
&{\cal P}_{{\rm{sr}}}^{{\rm{LB}}}\left( {{P_{\rm{r}}},{{\cal C}_x}} \right) = 1 - \frac{{{e^{ - \frac{{{\Psi _{\rm{r}}}\left( {{{\cal C}_x}} \right)}}{{{P_{\rm{s}}}{\theta _{{\rm{sr}}}}}}}}}}{{\Gamma \left( {{m_{{\rm{rr}}}}} \right)\theta _{{\rm{rr}}}^{{m_{{\rm{rr}}}}}}}\sum\limits_{m = 0}^{{m_{{\rm{rd}}}} - 1} {\frac{1}{{\Gamma \left( {m + 1} \right)}}} \nonumber \\
&\times {\left( {\frac{{{\Psi _{\rm{r}}}\left( {{{\cal C}_x}} \right)}}{{{P_{\rm{s}}}{\theta _{{\rm{sr}}}}}}} \right)^m}\hspace{-3pt}\int_0^\infty  \hspace{-3pt}{{x^{{m_{{\rm{rr}}}} - 1}}} {\left( {{P_{\rm{r}}}x + 1} \right)^m}{e^{ - \left( {\frac{{{P_{\rm{r}}}{\Psi _{\rm{r}}}\left( {{{\cal C}_x}} \right)}}{{{P_{\rm{s}}}{\theta _{{\rm{sr}}}}}} + \frac{1}{{{\theta _{{\rm{rr}}}}}}} \right)x}}dx\cdot
\end{align} 
By using the binomial theorem  ${\left( {{P_{\rm{r}}}x + 1} \right)^m} = \sum\limits_{k = 0}^m \binom mk {\left( {{P_{\rm{r}}}x} \right)^k}$, the result is obtained by solving the integral in \eqref{out_S-R_temp} by \cite[3.381-4]{gradstein1980tables}. 
\end{proof}
\vspace{10pt}
 \subsubsection{Upper Bound over Rayleigh Fading}
In the case of  Rayleigh fading, the ${\rm S - R}$ hop outage probability is simplified as
\begin{equation}\label{p_out_sr_expectation_rayleigh}
\PsrR\left( {{P_{\rm{r}}},{{\cal C}_x}} \right) = 1 - \mathbb{E}{_{{g_{\rm{rr}}}}}\left\{ {{e^{ - \frac{{\left( {{P_{\rm{r}}}{g_{\rm{rr}}} + 1} \right)}}{{{P_{\rm{s}}}{\pi _{\rm{sr}}}}}\Psi_r \left( {\frac{{{P_{\rm{r}}}{g_{\rm{rr}}}{{\cal C}_x}}}{{{P_{\rm{r}}}{g_{\rm{rr}}} + 1}}} \right)}}} \right\}\cdot
\end{equation}
Unfortunately, there is no closed-form expression for this expectation except at $\mathcal{C}_x=0$, which gives the known \ac{PGS} outage probability given  as
\begin{align}\label{p_out_sr_C_0_rayleigh}
\PsrR\left( {{P_{\rm{r}}},0} \right) = 1 - \frac{{{P_{\rm{s}}}{\pi _{{\rm{sr}}}}{e^{ - \frac{\eta }{{{P_{\rm{s}}}{\pi _{{\rm{sr}}}}}}}}}}{{{P_{\rm{s}}}{\pi _{{\rm{sr}}}}+{P_{\rm{r}}}{\pi _{{\rm{rr}}}}\eta  }},
\end{align}
where $\eta  = {2^r} - 1$. Otherwise, we resort to obtain an upper bound  as follows.
\begin{proposition}\label{prop_1}
The exponential term inside the expectation operator in \eqref{p_out_sr_expectation_rayleigh} is  convex  in $g_{\rm{rr}}$.
\end{proposition}
\begin{proof}
The proof is given in Appendix \ref{prop_1_proof}.
\end{proof}
Therefore,  an upper bound on the ${\rm S -  R}$ link outage probability is given as follows.
\begin{lemma}\label{OP_S-R_UB} 
The ${\rm S -  R}$ link outage probability, over Rayleigh fading, is upper-bounded by
\vspace{5pt}
\begin{align}
\PsrR \left( {{P_{\rm{r}}},{{\cal C}_x}} \right) \leq 1 - {e^{ - \frac{{ {{P_{\rm{r}}}{\pi _{\rm{rr}}} + 1} }}{{{P_{\rm{s}}}{\pi _{\rm{sr}}}}}\Psi_r \left( \alpha \mathcal{C}_x \right)}}\buildrel \Delta \over =\PsrR^{\rm{UB}} \left( {{P_{\rm{r}}},{{\cal C}_x}} \right),
\end{align}
where  $\alpha  = \frac{{{P_{\rm{r}}}{\pi _{{\rm{rr}}}}}}{{{P_{^{\rm{r}}}}{\pi _{{\rm{rr}}}} + 1}}$.
\end{lemma}
\begin{proof}
First, we follow Proposition \ref{prop_1}, then by applying the Jensen's inequality to the expectation in \eqref{p_out_sr_expectation_rayleigh}, we obtain the given outage probability upper bound of the ${\rm S - R}$ link.
\end{proof}
\subsection{Outage Probability of ${\rm R - D}$ Link}
Similarly, let $r$ (bits/sec/Hz) denote the target rate of the ${\rm R - D}$ link, then its outage probability is defined as
 \begin{equation}\label{p_out_rd_prob_m}
\Prd\left( { {P_{\rm{r}}},{{\cal C}_x}} \right) = \mathbb{P}\left\{ {{{R}_{\rm{rd}}}\left( { {P_{\rm{r}}},{{\cal C}_x}} \right) < r} \right\}\cdot
 \end{equation}
 Then, the outage probability of the ${\rm R - D}$ link can be obtained from the following result.
 \begin{theorem}\label{OP_R-D_exact}
 The ${\rm R - D}$ link outage probability with a target rate of $r$ is expressed as
 \begin{align}
&{{\cal P}_{\rm{rd}}}\left( {{P_{\rm{r}}},{{\cal C}_x}} \right) = 1 - \frac{{{e^{ -\frac{\Psi_r\left( { {{\cal C}_x}} \right)}{P_{\rm{r}}\theta_{\rm{rd}}\left(1-\mathcal{C}_x^2\right)}}}}}{{\Gamma \left( {{m_{{\rm{sd}}}}} \right)\theta _{{\rm{sd}}}^{{m_{{\rm{sd}}}}}}}\nonumber \\
&\times\sum\limits_{m = 0}^{{m_{{\rm{rd}}}} - 1} {\sum\limits_{k = 0}^m \binom mk \frac{{P_{\rm{s}}^k\Gamma \left( {k + {m_{{\rm{sd}}}}} \right)\left(\frac{\Psi_r\left( { {{\cal C}_x}} \right)}{P_{\rm{r}}\theta_{\rm{rd}}\left(1-\mathcal{C}_x^2\right)}\right)^m}}{{\Gamma \left( {m + 1} \right){{\left( {\frac{{P_{\rm{s}}}\Psi_r\left( { {{\cal C}_x}} \right)}{P_{\rm{r}}\theta_{\rm{rd}}\left(1-\mathcal{C}_x^2\right)} + \frac{1}{{{\theta _{{\rm{sd}}}}}}} \right)}^{k + {m_{{\rm{sd}}}}}}}}}\cdot
\end{align}
 
 \end{theorem}
 \begin{proof}
 The proof  follows similar steps as in Lemma \ref{OP_S-R_int} and the ${\rm S -  R}$ outage probability lower bound over Rayleigh fading in Lemma \ref{OP_S-R_LB}.
 \end{proof}

For Rayleigh fading,  the outage probability of the ${\rm R - D}$ link can be obtained from the following corollary.
 \begin{corollary}\label{OP_R-D_rayleigh}
The outage probability of the ${\rm R - D}$ link, over Rayleigh fading, is expressed  as
\begin{equation}
\PrdR\left( {{P_{\rm{r}}},{{\cal C}_x}} \right) = 1 - \frac{{{e^{ - \frac{{\Psi_r \left( {{{\cal C}_x}} \right)}}{{{P_{\rm{r}}}{\pi _{\rm{rd}}}\left( {1 - {\cal C}_x^2} \right)}}}}}}{{{P_{\rm{s}}}{\pi _{\rm{sd}}}\frac{{\Psi_r \left( {{{\cal C}_x}} \right)}}{{{P_{\rm{r}}}{\pi _{\rm{rd}}}\left( {1 - {\cal C}_x^2} \right)}} + 1}}\cdot \label{eq:RD_exact_rayleigh}
\end{equation}
 \end{corollary} 
  From Corollary \ref{OP_R-D_rayleigh}, it can be noticed that, for the \ac{PGS} case, i.e., $\mathcal{C}_x=0$, Eq. \eqref{eq:RD_exact_rayleigh} yields the known expression for \ac{PGS} \cite{kwon2010optimal} as 
\begin{equation}\label{p_out_rd_C_0}
\PrdR\left( {{P_{\rm{r}}},0} \right) = 1 - \frac{{{P_{\rm{r}}}{\pi _{{\rm{rd}}}}{e^{ - \frac{\eta }{{{P_{\rm{r}}}{\pi _{{\rm{rd}}}}}}}}}}{{{P_{\rm{r}}}{\pi _{{\rm{rd}}}} + {P_{\rm{s}}}{\pi _{{\rm{sd}}}}\eta }}\cdot
\end{equation}
Also, for the maximally improper case, i.e., $\mathcal{C}_x=1$, it yields
\begin{equation}
\PrdR\left( {{P_{\rm{r}}},1} \right) = \mathop {\lim }\limits_{{{\cal C}_x} \to 1} \PrdR\left( {{P_{\rm{r}}},{{\cal C}_x}} \right) = 1 - \frac{{{e^{ - \frac{\gamma }{{2{P_{\rm{r}}}{\pi _{\rm{rd}}}}}}}}}{{\frac{{\gamma {P_{\rm{s}}}{\pi _{\rm{sd}}}}}{{2{P_{\rm{r}}}{\pi _{\rm{rd}}}}} + 1}}\cdot
\end{equation} 
\subsection{End-to-End Outage Probability}
The exact end-to-end outage performance can be obtained from the following theorem.
\begin{theorem}\label{OP_E-E_int}
The exact end-to-end outage probability with a target rate $r$  can be numerically calculated from
\begin{align}\label{P_E_E_not_stable}
&\Pout = 1- \frac{{{{e^{ -\frac{\Psi_r\left( { {{\cal C}_x}} \right)}{P_{\rm{r}}\theta_{\rm{rd}}\left(1-\mathcal{C}_x^2\right)}}}}}}{{{\Gamma \left( {{m_{{\rm{sd}}}}} \right)}}{\Gamma \left( {{m_{\rm{rr}}}} \right)\Gamma \left( {{m_{{\rm{sr}}}}} \right)\theta _{{\rm{sd}}}^{{m_{{\rm{sd}}}}}\theta _{\rm{rr}}^{{m_{\rm{rr}}}}}} \nonumber \\
&\times\sum\limits_{m = 0}^{{m_{{\rm{rd}}}} - 1} {\sum\limits_{k = 0}^m \binom mk \frac{{P_{\rm{s}}^k\Gamma \left( {k + {m_{{\rm{sd}}}}} \right)\left(\frac{\Psi_r\left( { {{\cal C}_x}} \right)}{P_{\rm{r}}\theta_{\rm{rd}}\left(1-\mathcal{C}_x^2\right)}\right)^m}}{{\Gamma \left( {m + 1} \right){{\left( {\frac{{P_{\rm{s}}}\Psi_r\left( { {{\cal C}_x}} \right)}{P_{\rm{r}}\theta_{\rm{rd}}\left(1-\mathcal{C}_x^2\right)} + \frac{1}{{{\theta _{{\rm{sd}}}}}}} \right)}^{k + {m_{{\rm{sd}}}}}}}}} \nonumber \\
&\times \int\limits_0^\infty  {{x^{{m_{\rm{rr}}} - 1}}\Gamma \left( {{m_{{\rm{sr}}}},\frac{{\left( {{P_{\rm{r}}}x + 1} \right)}}{{{P_{\rm{s}}}{\theta _{{\rm{sr}}}}}}\Psi_r \left( {\frac{{{P_{\rm{r}}}x{{\cal C}_x}}}{{{P_{\rm{r}}}x + 1}}} \right)} \right){e^{ - \frac{x}{{{\theta _{\rm{rr}}}}}}}} dx \cdot
\end{align}
\end{theorem} 
\begin{proof}
Based on the derived expressions of the outage probability for ${\rm S - R}$ and  ${\rm R - D}$ links from Lemma \ref{OP_S-R_int} and Theorem \ref{OP_R-D_exact}, respectively, and by direct substitution in \eqref{p_out_overall}, we obtain the result.
\end{proof}
Unfortunately, there is no closed form solution for the end-to-end outage probability in Theorem \ref{OP_E-E_int}. Therefore, we resort to obtain an expression for the  lower bound of the end-to-end outage probability in the Nakagami-$m$ fading, in addition to an upper bound expression for the special case of Rayleigh fading.  These bounds have a significantly reduced computational complexity and an enhanced numerical stability than the exact expression in \eqref{P_E_E_not_stable}.\\
\subsubsection{Lower Bound}
The following theorem provides a lower bound for the   end-to-end outage probability.
\begin{theorem}\label{out_E_E_LB}
The end-to-end outage probability  can be lower-bounded as
\begin{align}\label{out_E_E_LB_nt}
&{\cal P}_{{\rm{E - E}}}\left( {{P_{\rm{r}}},{{\cal C}_x}} \right) \geq 1 - \frac{{{e^{ - {\Psi _r}\left( {{{\cal C}_x}} \right)\left( {\frac{1}{{{P_{\rm{s}}}{\theta _{{\rm{sr}}}}}} + \frac{1}{{{P_{\rm{r}}}{\theta _{{\rm{rd}}}}\left( {1 - {\cal C}_x^2} \right)}}} \right)}}}}{{\Gamma \left( {{m_{{\rm{sd}}}}} \right)\Gamma \left( {{m_{{\rm{rr}}}}} \right)\theta _{{\rm{sd}}}^{{m_{{\rm{sd}}}}}\theta _{{\rm{rr}}}^{{m_{{\rm{rr}}}}}}} \nonumber\\[3pt]
&\times\sum\limits_{m = 0}^{{m_{{\rm{sr}}}} - 1} \sum\limits_{m' = 0}^{{m_{{\rm{rd}}}} - 1} \sum\limits_{k = 0}^m \sum\limits_{k' = 0}^{m'}  \!\binom mk \binom {m'}{k'} \nonumber \\[3pt]
& \times\frac{{\frac{{P_{\rm{r}}^{k - m'}\Gamma \left( {k + {m_{{\rm{rr}}}}} \right)\Gamma \left( {k' + {m_{{\rm{sd}}}}} \right){\Psi _r^{m + m'}}\left( {{{\cal C}_x}} \right) }}{{P_{\rm{s}}^{m - k'}\Gamma \left( {m + 1} \right)\Gamma \left( {m' + 1} \right)\theta _{{\rm{sr}}}^m{{ {{\theta _{{\rm{rd}}}^{m'}}\left( {1 - {\cal C}_x^2} \right)} }^{m'}}}}}}{{{{\left( {\frac{{{P_{\rm{r}}}{\Psi _r}\left( {{{\cal C}_x}} \right)}}{{{P_{\rm{s}}}{\theta _{{\rm{sr}}}}}} + \frac{1}{{{\theta _{{\rm{rr}}}}}}} \right)}^{k + {m_{{\rm{rr}}}}}}{{\left( {\frac{{{P_{\rm{s}}}{\Psi _r}\left( {{{\cal C}_x}} \right)}}{{{P_{\rm{r}}}{\theta _{{\rm{rd}}}}\left( {1 - {\cal C}_x^2} \right)}} + \frac{1}{{{\theta _{{\rm{sd}}}}}}} \right)}^{k' + {m_{{\rm{sd}}}}}}}}  \nonumber \\[3pt]
&  \!\buildrel \Delta \over = \! {\cal P}_{{\rm{E - E}}}^{{\rm{LB}}}\left( {{P_{\rm{r}}},{{\cal C}_x}} \right)\cdot
\end{align}
\end{theorem}
\begin{proof}
The result follows directly from  \eqref{p_out_overall}, Lemma \ref{OP_S-R_LB} and Theorem \ref{OP_R-D_exact}. 
\end{proof}
\begin{corollary}\label{out_E_E_Exact_RF}
The lower bound in Theorem \ref{out_E_E_LB} at $\mathcal{C}_x=0$ reduces to the exact expression of the end-to-end outage probability for the \ac{PGS} case.  Specifically, for the \ac{PGS} case over Rayleigh fading, we have the exact expression for the end-to-end outage probability as
\begin{equation}\label{e2e_proper}
{{\cal P}_{{\rm{E-E,RF}}}}\left( {{P_{\rm{r}}},0} \right) = 1 - \frac{{{P_{\rm{s}}}{P_{\rm{r}}}{\pi _{{\rm{sr}}}}{\pi _{{\rm{rd}}}}{e^{ - \eta \left( {\frac{1}{{{P_{\rm{s}}}{\pi _{{\rm{sr}}}}}} + \frac{1}{{{P_{\rm{r}}}{\pi _{{\rm{rd}}}}}}} \right)}}}}{{\left( {{P_{\rm{s}}}{\pi _{{\rm{sr}}}} + {P_{\rm{r}}}{\pi _{{\rm{rr}}}}\eta } \right)\left( {{P_{\rm{r}}}{\pi _{{\rm{rd}}}} + {P_{\rm{s}}}{\pi _{{\rm{sd}}}}\eta } \right)}}\cdot
\end{equation} 
\end{corollary}
\begin{proof}
For $\mathcal{C}_x=0$, the ${\rm S - R}$ outage probability lower bound in Lemma  \ref{OP_S-R_LB} reduces to the exact expression and then the result follows directly.
\end{proof}
\vspace{6pt}
\subsubsection{Upper Bound over Rayleigh Fading}
From a system design prospective, it is typically more beneficial to have an upper bound than a lower bound for the end-to-end outage probability. Therefore, we state next the upper bound of the end-to-end outage probability over Rayleigh fading in Theorem \ref{theorem_e2e_outage}.
\begin{theorem}\label{theorem_e2e_outage}
The end-to-end outage probability over Rayleigh fading can be upper-bounded by 
\begin{align}\label{p_out_UB}
\vspace{10pt}
{{\cal P}_{{\rm{E-E,RF}}}}  \left( {{P_{\rm{r}}},{{\cal C}_x}} \right) &\leq 1 - \frac{{{e^{ -\Big(\frac{{\Psi_r \left( {{{\cal C}_x}} \right)}}{{{P_{\rm{r}}}{\pi _{\rm{rd}}}\left( {1 - {\cal C}_x^2} \right)}}+ \frac{{{{P_{\rm{r}}}{\pi _{\rm{rr}}} + 1} }}{{{P_{\rm{s}}}{\pi _{\rm{sr}}}}}\Psi_r \left( \alpha \mathcal{C}_x \right)\Big)}}}}{{{P_{\rm{s}}}{\pi _{\rm{sd}}}\frac{{{\Psi_r}\left( {{{\cal C}_x}} \right)}}{{{P_{\rm{r}}}{\pi _{\rm{rd}}}\left( {1 - {\cal C}_x^2} \right)}} + 1}} \nonumber \\
&\buildrel \Delta \over ={{\cal P}^{\rm{UB}}_{{\rm{E-E,RF}}}}  \left( {{P_{\rm{r}}},{{\cal C}_x}} \right) \cdot
\end{align}
\end{theorem}
\newpage
\begin{proof}
It follows directly from \eqref{p_out_overall}, Lemma \ref{OP_S-R_UB} and Corollary \ref{OP_R-D_rayleigh}.
\end{proof}
\vspace{10pt}
\textit{Asymptotic Analysis:} 
For maximally \ac{IGS}, we obtain the upper bound of the end-to-end outage probability from the following corollary.
\begin{corollary}\label{corollary1}
When the relay node uses maximally \ac{IGS}, the end-to-end outage probability, over Rayleigh fading, can be upper-bounded by
\begin{equation}
\mathop {\lim }\limits_{{{\cal C}_x} \to 1} {{\cal P}^{\rm{UB}}_{{\rm{E-E,RF}}}} = 1 - \frac{{2{P_{\rm{r}}}{\pi _{{\rm{rd}}}}{e^{ - \left( {\frac{\gamma }{{2{P_{\rm{r}}}{\pi _{{\rm{rd}}}}}} + \frac{{\left( {{P_{\rm{r}}}{\pi _{{\rm{rr}}}} + 1} \right)}}{{{P_{\rm{s}}}{\pi _{{\rm{sr}}}}}}{\Psi_r }\left(  \alpha \mathcal{C}_x\right)} \right)}}}}{{2{P_{\rm{r}}}{\pi _{{\rm{rd}}}} + \gamma {P_{\rm{s}}}{\pi _{{\rm{sd}}}}}}\cdot
\end{equation}
\end{corollary}  
In order to evaluate the performance of the end-to-end outage probability upper bound with respect to \ac{RSI} when using maximally \ac{IGS} at the relay transmitter, we state the following theorem.
\begin{theorem}
In the limiting case where $\pi_{\rm{rr}} \rightarrow \infty$ over Rayleigh fading with a fixed relay transmit power $P_{\rm{r}}$, the exact end-to-end outage probability for the \ac{PGS} case ${{\cal P}_{{\rm{E-E,RF}}}}\left( {{P_{\rm{r}}},0} \right) \rightarrow 1$, while the upper bound for the end-to-end outage probability for the maximally \ac{IGS} case ${{\cal P}^{\rm{UB}}_{{\rm{E-E,RF}}}}\left( {{P_{\rm{r}}},1} \right) \rightarrow {K}$, where

\begin{equation}\label{p_out_pirr_inf}
{K} = 1 - \frac{{2{P_{\rm{r}}}{\pi _{{\rm{rd}}}}{e^{ - \left( {\frac{\gamma }{{{2P_{\rm{r}}}{\pi _{{\rm{rd}}}}}} + \frac{\gamma }{{{P_{\rm{s}}}{\pi _{{\rm{sr}}}}}}} \right)}}}}{{2{P_{\rm{r}}}{\pi _{{\rm{rd}}}} + \gamma {P_{\rm{s}}}{\pi _{{\rm{sd}}}}}}\cdot
\end{equation} 
\end{theorem}
\begin{proof}
The result is obtained from Corollary \ref{out_E_E_Exact_RF}, \ref{corollary1} and taking the limit at $\pi_{\rm{rr}} \rightarrow \infty$.
\end{proof}

Interestingly, different from the \ac{PGS} case, the maximally \ac{IGS} introduces immunity against high \ac{RSI} and achieves less outage probability with a constant upper bound \eqref{p_out_pirr_inf}, which depends on the quality of both $\rm{S-R}$ and $\rm{R-D}$ links, in addition to the target rate. 
\section{Ergodic Rate Performance}\label{sec:analysis}
In this section, we evaluate the ergodic rate performance of the canonical cooperative setting depicted in Fig. \ref{sysmodfig} when \ac{IGS} is allowed at the relay under Nakagami-$m$ fading. We also present simplified expressions for Rayleigh fading.

The end-to-end ergodic rate of the \ac{FDR} system can be calculated as 
\begin{align}\label{ER_E-E_prob}
{{\cal R}_{{\rm{E - E}}}} =\mathbb E\left\{ {\min \left( {{R_{{\rm{sr}}}},{R_{{\rm{rd}}}}} \right)} \right\} = \hspace{-5pt}\int\limits_0^\infty  {\mathbb P\left\{ {\min \left\{ {{R_{{\rm{sr}}}},{R_{{\rm{rd}}}}} \right\} \ge r} \right\}} \;dr\cdot
\end{align}
The last complementary \ac{CDF} integral form for the statistical expectation follows from the fact that ${\min \left\{ {{R_{{\rm{sr}}}},{R_{{\rm{rd}}}}} \right\}}$ is a non-negative \ac{RV}. Hence the exact end-to-end ergodic rate can be numerically computed from the following result.
\begin{theorem}\label{ER_E-E_int}
The exact end-to-end ergodic rate  can be numerically calculated from
\footnotesize
\begin{align}\label{R_E_E_not_stable}
&{\cal R_{{\rm{E - E}}}} = \sum\limits_{m = 0}^{{m_{{\rm{rd}}}} - 1} \sum\limits_{k = 0}^m  \binom mk  \frac{{{{P_{\rm{s}}^k}}\Gamma \left( {k + {m_{{\rm{sd}}}}} \right)}}{{\Gamma \left( {m + 1} \right)}{}}  \nonumber \\
&\times \hspace{-2pt} \int\limits_0^\infty  \int\limits_0^\infty \hspace{-2pt} \frac{{{x^{{m_{\rm{rr}}} - 1}}\Gamma \left( {{m_{{\rm{sr}}}},\frac{{\left( {{P_{\rm{r}}}x + 1} \right)}}{{{P_{\rm{s}}}{\theta _{{\rm{sr}}}}}}\Psi_r \left( {\frac{{{P_{\rm{r}}}x{{\cal C}_x}}}{{{P_{\rm{r}}}x + 1}}} \right)} \right)\left(\frac{\Psi_r\left( { {{\cal C}_x}} \right)}{P_{\rm{r}}\theta_{\rm{rd}}\left(1-\mathcal{C}_x^2\right)}\right)^m}}{{\Gamma \left( {{m_{{\rm{sd}}}}} \right)\Gamma \left( {{m_{\rm{rr}}}} \right)\Gamma \left( {{m_{{\rm{sr}}}}} \right)\theta _{{\rm{sd}}}^{{m_{{\rm{sd}}}}}\theta _{\rm{rr}}^{{m_{\rm{rr}}}}{{\left( {\frac{{{P_{\rm{s}}}\Psi_r \left( {{{\cal C}_x}} \right)}}{{{P_{\rm{r}}}{\theta _{{\rm{rd}}}}\left( {1 - {\cal C}_x^2} \right)}} + \frac{1}{{{\theta _{{\rm{sd}}}}}}} \right)}^{k + {m_{{\rm{sd}}}}}}}} \nonumber \\
&\times {e^{ - \left( {\frac{{\Psi_r \left( {{{\cal C}_x}} \right)}}{{{P_{\rm{r}}}{\theta _{{\rm{rd}}}}\left( {1 - {\cal C}_x^2} \right)}} + \frac{x}{{{\theta _{\rm{rr}}}}}} \right)}}\;dxdr\cdot 
\end{align}
\normalsize
\end{theorem} 
\begin{proof}
The result is obtained directly from \eqref{ER_E-E_prob}. The expression of the integrand represents the complement of the end-to-end outage probability which is  obtained in Theorem~ \ref{OP_E-E_int}.
\end{proof}
Unfortunately, this double integral\footnote{The direct numerical computation of this double integral needs a careful selection of the upper limit value to avoid overflow. } does not have a closed form solution in the general case. In what follows, we alternatively derive an upper bound on the end-to-end ergodic rate performance in the Nakagami-$m$ fading. Additionally, we obtain a lower bound in the Rayleigh fading special case. The derived bounds are computationally simpler and numerically more stable than the exact form in \eqref{R_E_E_not_stable}. 
\subsubsection{Upper Bound}
An upper bound for the end-to-end ergodic rate is presented in the following result.
\begin{theorem}\label{E_E_rate_ergodic_UB}
The end-to-end ergodic rate can be upper-bounded as \eqref{eq_very_big}
\begin{figure*}
\small
\begin{align}\label{eq_very_big}
{\cal R}_{{\rm{E - E}}} &\leq \frac{1}{{\Gamma \left( {{m_{{\rm{sd}}}}} \right)\Gamma \left( {{m_{{\rm{rr}}}}} \right)\theta _{{\rm{sd}}}^{{m_{{\rm{sd}}}}}\theta _{{\rm{rr}}}^{{m_{{\rm{rr}}}}}\log \left( 2 \right)}} \sum\limits_{m = 0}^{{m_{{\rm{sr}}}} - 1} {\sum\limits_{m' = 0}^{{m_{{\rm{rd}}}} - 1} {\sum\limits_{k = 0}^m {\sum\limits_{k' = 0}^{m'} \binom mk \binom {m'}{k'} \frac{{P_{\rm{r}}^{k - m'}\Gamma \left( {k + {m_{{\rm{rr}}}}} \right)\Gamma \left( {k' + {m_{{\rm{sd}}}}} \right)}}{{P_{\rm{s}}^{m - k'}\Gamma \left( {m + 1} \right)\Gamma \left( {m' + 1} \right)\theta _{{\rm{sr}}}^m{{\left( {{\theta _{{\rm{rd}}}}\left( {1 - {\cal C}_x^2} \right)} \right)}^{m'}}}}} }  }\nonumber\\[0.3cm]
& \times \Bigg(\sum\limits_{i = 1}^2 {{{\cal D}_i}{\Omega ^{ - \left( {m + m'} \right)}}\Gamma \left( {m + m' + 1} \right){\Xi _{m + m' + 1}}\left( {\left( {1 + {{\left( { - 1} \right)}^i}{{\cal C}_x}} \right)\Omega } \right)}\nonumber \\[0.3cm]
&\hspace{20pt} + \sum\limits_{j = 1}^{k + {m_{{\rm{rr}}}}} {\zeta _j P_{\rm{r}}^{ - j}P_{\rm{s}}^j\theta _{{\rm{sr}}}^j{\Omega ^{\left( {j - 1} \right) - \left( {m + m'} \right)}}\Gamma \left( {m + m' + 1} \right)U\left( {j,j - \left( {m + m'} \right),\frac{1}{{{P_{\rm{r}}}{\theta _{{\rm{rr}}}}}} + \frac{{{P_{\rm{s}}}{\theta _{{\rm{sr}}}}}}{{{P_r}^2{\theta _{{\rm{rd}}}}{\theta _{{\rm{rr}}}}\left( {1 - {\cal C}_x^2} \right)}}} \right)} \nonumber \\[0.3cm]
&\hspace{20pt}+\sum\limits_{l = 1}^{k' + {m_{{\rm{sd}}}}} {{\xi _l}P_{\rm{s}}^{ - l}P_{\rm{r}}^l\theta _{{\rm{rd}}}^l{{\left( {1 - {\cal C}_x^2} \right)}^l}{\Omega ^{\left( {l - 1} \right) - \left( {m + m'} \right)}}\Gamma \left( {m + m' + 1} \right)U\left( {l,l - \left( {m + m'} \right),\frac{1}{{{P_{\rm{s}}}{\theta _{{\rm{sd}}}}}} + \frac{{{P_{\rm{r}}}{\theta _{{\rm{rd}}}}\left( {1 - {\cal C}_x^2} \right)}}{{P_{\rm{s}}^2{\theta _{{\rm{sr}}}}{\theta _{{\rm{sd}}}}}}} \right)}  \Bigg)   \buildrel \Delta \over = {\cal R}_{{\rm{E - E}}}^{{\rm{UB}}},
\end{align}
where $\Omega  = \frac{1}{{{P_{\rm{s}}}{\theta _{{\rm{sr}}}}}} + \frac{1}{{{P_{\rm{r}}}{\theta _{{\rm{rd}}}}\left( {1 - {\cal C}_x^2} \right)}}$, ${\lambda _i} = \mathop {\lim }\limits_{{\Psi _r} \to  - \left( {1 + {{\left( { - 1} \right)}^i}{{\cal C}_x}} \right)} \left( {{\Psi _r} + \left( {1 + {{\left( { - 1} \right)}^i}{{\cal C}_x}} \right)} \right){\cal F}\left( {{\Psi _r}} \right)$, \\[0.5cm] ${\zeta _j} \hspace{-2pt}= \hspace{-2pt}\frac{{\mathop {\lim }\limits_{{\Omega _r} \to  - \frac{{{P_{\rm{s}}}{\theta _{{\rm{sr}}}}}}{{{P_{\rm{r}}}{\theta _{{\rm{rr}}}}}}} \frac{{{d^{k + {m_{{\rm{rr}}}} - j}}}}{{d\Psi _r^{k + {m_{{\rm{rr}}}} - j}}}\left( {{{\left( {\frac{{{P_{\rm{r}}}{\Psi _r}}}{{{P_{\rm{s}}}{\theta _{{\rm{sr}}}}}} + \frac{1}{{{\theta _{{\rm{rr}}}}}}} \right)}^{k + {m_{{\rm{rr}}}}}}{\cal F}\left( {{\Psi _r}} \right)} \right)}}{{{{\left( {\frac{{{P_{\rm{r}}}}}{{{P_{\rm{s}}}{\theta _{{\rm{sr}}}}}}} \right)}^{k + {m_{{\rm{rr}}}} - j}}\left( {k + {m_{{\rm{rr}}}} - j} \right)!}}$, ${\xi _l} \hspace{-2pt}=\hspace{-2pt} \frac{{\mathop {\lim }\limits_{{\Omega _r} \to  - \frac{{{P_{\rm{r}}}{\theta _{{\rm{rd}}}}\left( {1 - {\cal C}_x^2} \right)}}{{{P_{\rm{s}}}{\theta _{{\rm{sd}}}}}}} \frac{{{d^{k' + {m_{{\rm{sd}}}} - l}}}}{{d\Psi _r^{k' + {m_{{\rm{sd}}}} - l}}}\left( {{{\left( {\frac{{{P_{\rm{s}}}{\Psi _r}}}{{{P_{\rm{r}}}{\theta _{{\rm{rd}}}}\left( {1 - {\cal C}_x^2} \right)}} + \frac{1}{{{\theta _{{\rm{sd}}}}}}} \right)}^{k' + {m_{{\rm{sd}}}}}}{\cal F}\left( {{\Psi _r}} \right)} \right)}}{{{{\left( {\frac{{{P_{\rm{s}}}}}{{{P_{\rm{r}}}{\theta _{{\rm{rd}}}}\left( {1 - {\cal C}_x^2} \right)}}} \right)}^{k' + {m_{{\rm{sd}}}} - l}}\left( {k' + {m_{{\rm{sd}}}} - l} \right)!}}$\\[0.5cm] and ${\cal F}\left( {{\Psi _r}} \right) = \frac{{\left( {{\Psi _r} + 1} \right)}}{{\prod\limits_{i = 1}^2 {\left( {{\Psi _r} + \left( {1 + {{\left( { - 1} \right)}^i}{{\cal C}_x}} \right)} \right)} {{\left( {\frac{{{P_{\rm{r}}}{\Psi _r}}}{{{P_{\rm{s}}}{\theta _{{\rm{sr}}}}}} + \frac{1}{{{\theta _{{\rm{rr}}}}}}} \right)}^{k + {m_{{\rm{rr}}}}}}{{\left( {\frac{{{P_{\rm{s}}}{\Psi _r}}}{{{P_{\rm{r}}}{\theta _{{\rm{rd}}}}\left( {1 - {\cal C}_x^2} \right)}} + \frac{1}{{{\theta _{{\rm{sd}}}}}}} \right)}^{k' + {m_{{\rm{sd}}}}}}}}\cdot$
\normalsize
\\\\\vspace*{10pt}
\hrulefill
\end{figure*}
\end{theorem}
\begin{proof}
We define an upper bound for the end-to-end ergodic rate of the FDR system as 

\begin{equation}\label{R_E_E_int}
{{\cal R}_{{\rm{E-E}}}}\le \int\limits_0^\infty 1-{\cal P}_{{\rm{E - E}}}^{{\rm{LB}}}\left( {{P_{\rm{r}}},{{\cal C}_x}} \left|r\right.\right) \;dr\buildrel \Delta \over={{\cal R}^{\rm{UB}}_{{\rm{E-E}}}}\left( {{P_{\rm{r}}},{{\cal C}_x}} \right),
\end{equation}
where ${\cal P}_{{\rm{E - E}}}^{{\rm{LB}}}$ is the lower bound of the outage probability which is obtained in Theorem \ref{out_E_E_LB}, which depends on the target rate $r$. By substituting \eqref{out_E_E_LB_nt} in \eqref{R_E_E_int}, we get the following integral expression for the upper bound of the end-to-end ergodic rate.
\footnotesize
\begin{align}
{\cal R}_{{\rm{E - E}}}^{{\rm{UB}}} &= \frac{1}{{\Gamma \left( {{m_{{\rm{sd}}}}} \right)\Gamma \left( {{m_{{\rm{rr}}}}} \right)\theta _{{\rm{sd}}}^{{m_{{\rm{sd}}}}}\theta _{{\rm{rr}}}^{{m_{{\rm{rr}}}}}}}\sum\limits_{m = 0}^{{m_{{\rm{sr}}}} - 1} \sum\limits_{m' = 0}^{{m_{{\rm{rd}}}} - 1} \sum\limits_{k = 0}^m \sum\limits_{k' = 0}^{m'} \binom mk \binom {m'}{k'} \nonumber \\[2pt]
&\times\frac{{P_{\rm{r}}^{k - m'}\Gamma \left( {k + {m_{{\rm{rr}}}}} \right)\Gamma \left( {k' + {m_{{\rm{sd}}}}} \right)}}{{P_{\rm{s}}^{m - k'}\Gamma \left( {m + 1} \right)\Gamma \left( {m' + 1} \right)\theta _{{\rm{sr}}}^m{{\left( {{\theta _{{\rm{rd}}}}\left( {1 - {\cal C}_x^2} \right)} \right)}^{m'}}}}    \nonumber \\[2pt]
&  \times\hspace{-1pt} \underbrace{\int\limits_0^\infty \hspace{-1pt} {\frac{{{e^{ - {\Psi _r}\left( {{{\cal C}_x}} \right)\left( {\frac{1}{{{P_{\rm{s}}}{\theta _{{\rm{sr}}}}}} + \frac{1}{{{P_{\rm{r}}}{\theta _{{\rm{rd}}}}\left( {1 - {\cal C}_x^2} \right)}}} \right)}}{{\left( {{\Psi _r}\left( {{{\cal C}_x}} \right)} \right)}^{m + m'}}}}{{{{\left( {\frac{{{P_{\rm{r}}}{\Psi _r}\left( {{{\cal C}_x}} \right)}}{{{P_{\rm{s}}}{\theta _{{\rm{sr}}}}}} + \frac{1}{{{\theta _{{\rm{rr}}}}}}} \right)}^{k + {m_{{\rm{rr}}}}}}{{\left( {\frac{{{P_{\rm{s}}}{\Psi _r}\left( {{{\cal C}_x}} \right)}}{{{P_{\rm{r}}}{\theta _{{\rm{rd}}}}\left( {1 - {\cal C}_x^2} \right)}} + \frac{1}{{{\theta _{{\rm{sd}}}}}}} \right)}^{k' + {m_{{\rm{sd}}}}}}}}} \;dr}_{\mathcal{I}}\hspace{-1pt}\cdot
\end{align}
\normalsize
By employing a change of variables in the integral $\mathcal{I}$ and using partial fraction decomposition, $\mathcal{I}$ can be expressed as
\begin{align}
{\cal I} = \frac{1}{{\log \left( 2 \right)}}\Bigg[& \sum\limits_{i = 1}^2 {\int\limits_0^\infty  {\frac{{{\lambda _i}\Psi _r^{^{m + m'}}{e^{ - \Omega {\Psi _r}}}}}{{{\Psi _r} + \left( {1 + {{\left( { - 1} \right)}^i}{{\cal C}_x}} \right)}}} } d{\Psi _r} \nonumber \\
&+ \sum\limits_{j = 1}^{k + {m_{{\rm{rr}}}}} {\int\limits_0^\infty  {\frac{{{\zeta _j}\Psi _r^{^{m + m'}}{e^{ - \Omega {\Psi _r}}}}}{{{{\left( {\frac{{{P_{\rm{r}}}{\Psi _r}}}{{{P_{\rm{s}}}{\theta _{{\rm{sr}}}}}} + \frac{1}{{{\theta _{{\rm{rr}}}}}}} \right)}^j}}}} } d{\Psi _r} \nonumber \\
&+ \sum\limits_{l = 1}^{k' + {m_{{\rm{sd}}}}} {\int\limits_0^\infty  {\frac{{{\xi _l}\Psi _r^{^{m + m'}}{e^{ - \Omega {\Psi _r}}}}}{{{{\left( {\frac{{{P_{\rm{s}}}{\Psi _r}}}{{{P_{\rm{r}}}{\theta _{{\rm{rd}}}}\left( {1 - {\cal C}_x^2} \right)}} + \frac{1}{{{\theta _{{\rm{sd}}}}}}} \right)}^l}}}d{\Psi _r}} }  \Bigg],
\end{align} 
which can be solved, using \cite[Eq. 3.353-5, 8.352-5]{gradstein1980tables} and \cite[Eq. 6.5.9]{prudnikov1998integrals} for the first integral, and \cite[Eq. 2.3.6-9]{prudnikov1998integrals} for the other integrals.
\end{proof}
\begin{remark}
Following Corollary \ref{out_E_E_Exact_RF}, the upper bound in Theorem \ref{E_E_rate_ergodic_UB} at $\mathcal{C}_x=0$ reduces to the exact expression of the end-to-end ergodic rate for the \ac{PGS} case.
\end{remark}
 \subsubsection{Lower Bound over Rayleigh Fading}
Similar to the discussion in the previous section, providing a lower bound for the ergodic rate offers further insights regarding the least performance merits \ac{IGS} is able to offer relative to \ac{PGS}. A lower bound for the ergodic rate of the \ac{FDR} system over Rayleigh fading channels is presented in the following theorem.
\begin{theorem}\label{ER_LB_R}
The end-to-end ergodic rate, over Rayleigh fading, can be lower-bounded by 
\begin{align}
&{\cal R}_{{\rm{E - E,RF}}} \geq 
\frac{{{P_{\rm{r}}}{\pi _{{\rm{rd}}}}\left( {1 - {\cal C}_x^2} \right)}}{{{P_{\rm{s}}}{\pi _{{\rm{sd}}}}\log \left( 2 \right)}} \nonumber \\
& \times \Bigg( \sum\limits_{i = 1}^2 {{\kappa _i}{\Xi _1}\left( {\frac{{{P_{\rm{r}}}{\pi _{{\rm{rr}}}}\left( {1 + {{\left( { - 1} \right)}^i}\alpha {{\cal C}_x}} \right)}}{{\alpha {P_{\rm{s}}}{\pi _{{\rm{sr}}}}}} + \frac{{\left( {1 + {{\left( { - 1} \right)}^i}\alpha {{\cal C}_x}} \right)}}{{{P_{\rm{r}}}{\pi _{{\rm{rd}}}}\left( {1 - {\cal C}_x^2} \right)}}} \right)}  \\ \nonumber &\hspace{0.9cm}+ {\kappa _3}{\Xi _1}\left( {\frac{{p_r^2{\pi _{{\rm{rd}}}}{\pi _{{\rm{rr}}}}\left( {1 - {\cal C}_x^2} \right)}}{{\alpha P_{\rm{s}}^2{\pi _{{\rm{sd}}}}{\pi _{{\rm{sr}}}}}} + \frac{1}{{{P_{\rm{s}}}{\pi _{{\rm{sd}}}}}}} \right) \Bigg)  \nonumber \\
&\buildrel \Delta \over = {\cal R}_{{\rm{E - E,RF}}}^{{\rm{LB}}},
\end{align}
\\
\text{where} ${\kappa _i} = \frac{{0.5{P_{\rm{s}}}{\pi _{{\rm{sd}}}}}}{{ {{P_{\rm{r}}}{\pi _{{\rm{rd}}}}\left( {1 - {\cal C}_x^2} \right) - {P_{\rm{s}}}{\pi _{{\rm{sd}}}}\left( {1 + {{\left( { - 1} \right)}^i}\alpha {{\cal C}_x}} \right)} }}$ and\\ ${\kappa _3} = \frac{{{P_{\rm{s}}}{\pi _{{\rm{sd}}}}\left( {{P_{\rm{s}}}{\pi _{{\rm{sd}}}} - {P_{\rm{r}}}{\pi _{{\rm{rd}}}}\left( {1 - {\cal C}_x^2} \right)} \right)}}{{\left( {{P_{\rm{r}}}{\pi _{{\rm{rd}}}}\left( {1 - {\cal C}_x^2} \right) - {P_{\rm{s}}}{\pi _{{\rm{sd}}}}\left( {1 - \alpha {{\cal C}_x}} \right)} \right)\left( {{P_{\rm{r}}}{\pi _{{\rm{rd}}}}\left( {1 - {\cal C}_x^2} \right) - {P_{\rm{s}}}{\pi _{{\rm{sd}}}}\left( {1 + \alpha {{\cal C}_x}} \right)} \right)}}\cdot$
\end{theorem}
\begin{proof}
First, we define a lower bound on the ergodic rate ~as 
\begin{equation}
{\cal R}_{{\rm{E - E,RF}}}^{{\rm{LB}}}=\int\limits_0^\infty 1- {{\cal P}^{\rm{UB}}_{{\rm{E-E,RF}}}} \left( {{P_{\rm{r}}},{{\cal C}_x}\left | r \right.} \right)\;dr,
\end{equation}
where ${{\cal P}^{\rm{UB}}_{{\rm{E-E,RF}}}}$ is the outage probability upper bound in the Rayleigh fading case which is stated in Theorem \ref{theorem_e2e_outage}. Hence, we ~get
\begin{align}
{\cal R}_{{\rm{E - E,RF}}}^{{\rm{LB}}} & = \int\limits_0^\infty  {\frac{{{e^{ - \left( {\frac{{\Psi_r \left( {{{\cal C}_x}} \right)}}{{{P_{\rm{r}}}{\pi _{\rm{rd}}}\left( {1 - {\cal C}_x^2} \right)}} + \frac{{\left( {{P_{\rm{r}}}{\pi _{\rm{rr}}} + 1} \right)}}{{{P_{\rm{s}}}{\pi _{\rm{sr}}}}}\Psi_r \left(  \alpha \mathcal{C}_x \right)} \right)}}}}{{{P_{\rm{s}}}{\pi _{\rm{sd}}}\frac{{\Psi_r \left( {{{\cal C}_x}} \right)}}{{{P_{\rm{r}}}{\pi _{\rm{rd}}}\left( {1 - {\cal C}_x^2} \right)}} + 1}}\;} d{r}\\[0.25cm] \nonumber &
 \overset{(a)}\ge \int\limits_0^\infty  {\frac{{{e^{ - \left( {\frac{{\Psi_r \left(  \alpha \mathcal{C}_x \right)}}{{{P_{\rm{r}}}{\pi _{\rm{rd}}}\left( {1 - {\cal C}_x^2} \right)}} + \frac{{\left( {{P_{\rm{r}}}{\pi _{\rm{rr}}} + 1} \right)}}{{{P_{\rm{s}}}{\pi _{\rm{sr}}}}}\Psi_r \left(  \alpha \mathcal{C}_x \right)} \right)}}}}{{{P_{\rm{s}}}{\pi _{\rm{sd}}}\frac{{\Psi_r \left(  \alpha \mathcal{C}_x \right)}}{{{P_{\rm{r}}}{\pi _{\rm{rd}}}\left( {1 - {\cal C}_x^2} \right)}} + 1}}\;} d{r},
\end{align}
where $(a)$ is obtained by replacing every $\Psi_r \left( {{{\cal C}_x}} \right)$ by $\Psi_r \left(  \alpha \mathcal{C}_x \right)$ as it can be readily verified that $\Psi_r \left( x\right)$ is a monotonically decreasing function in $x$ and $\alpha < 1$. Then, by changing of variables and performing partial fraction decomposition, we ~get
\begin{align}
&{\cal R}_{{\rm{E - E,RF}}}^{{\rm{LB}}} = \frac{{{P_{\rm{r}}}{\pi _{{\rm{rd}}}}\left( {1 - {\cal C}_x^2} \right)}}{{{P_{\rm{s}}}{\pi _{{\rm{sd}}}}\log \left( 2 \right)}} \nonumber \\
&\times \hspace{-3pt}\int\limits_0^\infty \hspace{-3pt} \left( {\sum\limits_{i = 1}^2 {\frac{{{\kappa _i}}}{{\left( {{\Psi _r} + \left( {1 + {{\left( { - 1} \right)}^i}\alpha {{\cal C}_x}} \right)} \right)}} + \frac{{{\kappa _3}}}{{\left( {{\Psi _r} + \frac{{{P_{\rm{r}}}{\pi _{{\rm{rd}}}}\left( {1 - {\cal C}_x^2} \right)}}{{{P_{\rm{s}}}{\pi _{{\rm{sd}}}}}}} \right)}}} }\hspace{-2pt} \right) \nonumber \\[0.2cm]
&\hspace{1cm}\times{e^{ - \left( {\frac{{{P_{\rm{r}}}{\pi _{{\rm{rr}}}}}}{{\alpha {P_{\rm{s}}}{\pi _{{\rm{sr}}}}}} + \frac{1}{{{P_{\rm{r}}}{\pi _{{\rm{rd}}}}\left( {1 - {\cal C}_x^2} \right)}}} \right)\Psi_r}}\;d\Psi_r,
\end{align}
which can be solved using \cite[Eq. 3.352-4]{gradstein1980tables} to give the lower bound. 
\end{proof}
\section{Improper Signaling Optimization}\label{sec: opt}
In this part, for Rayleigh fading, we optimize the parameters of the \ac{IGS} transmit signal to minimize the end-to-end outage probability upper bound given some boundaries for the optimization variables. First, we individually optimize the relay power/circularity coefficient, assuming the other variable is kept fixed. Second, we jointly optimize the power and circularity.   
\subsection{Individual Power and Circularity Coefficient Optimization}
In order to investigate the merits of \ac{IGS} over conventional \ac{PGS} in \ac{FDR} channels, we aim at finding the optimal circularity coefficient value that minimizes the end-to-end outage probability upper bound. Specifically, we aim at solving the following optimization problem:
\begin{align}
&\mathop {\min }\limits_{{{\cal C}_x}} \quad   \PoutR^{\rm{UB}}\left( {{P_{\rm{r}}},{{\cal C}_x}} \right) \nonumber \\
&\hspace{5pt}
\rm{s.t.} \quad\; 0 \leq {{\cal C}_x} \le 1\cdot 
\end{align}
In order to solve the optimization problem, we analyze the convexity properties of the objective function $\PoutR^{\rm{UB}}\left( {{P_{\rm{r}}},{{\cal C}_x}}\right)$ in \eqref{p_out_UB}. In general, the function is found to be non-convex due to the indefinite sign of the second derivative. However, other desirable properties that allow us to find the global optimal point are presented in the following theorem.
\begin{theorem}\label{theorem_unimodality}
The upper bound of the end-to-end outage probability, over Rayleigh fading, is either a monotonic or a unimodal function in each of its variables, ${P_{\rm{r}}}$ and ${{\cal C}_x}$, individually.

\end{theorem}
\begin{proof}
The proof is provided in Appendix \ref{theorem_unimodality_proof}.
\end{proof}
Since monotonicity and unimodality are special cases of quasi-convexity, such a result allows for the use of quasi-convex optimization algorithms. For instance, the optimal ${{\cal C}_x}$ can be numerically obtained using the well-known bisection method operating on its derivative given in Appendix B. Similar properties are shown for the individual power optimization problem. 
\subsection{Joint Power and Circularity Coefficient Optimization}
The power optimization problem in \ac{PGS} is formulated as
\begin{align}
&\mathop {\min }\limits_{{P_{\rm{r}}}} \quad   {\cal P}_{\rm E-E,RF}\left( {{P_{\rm{r}}},0} \right) \nonumber \\
&\hspace{5pt}
s.t. \quad \; 0 < {P_{\rm{r}}} \le {P_{\max }}\cdot
\end{align}
Also, for the \ac{IGS} case, the joint problem is given as follows:
\begin{align}\label{improper_opt}
&\mathop {\min }\limits_{{P_{\rm{r}}},{{\cal C}_x}} \quad   {\cal P}_{\rm E-E,RF}^{\rm UB}\left( {{P_{\rm{r}}},{{\cal C}_x}} \right)  \nonumber \\
&\hspace{5pt}
s.t. \quad  \;\; 0 < {P_{\rm{r}}} \le {P_{\max }},\nonumber\\
&\quad \quad \quad \;\;0 \leq {{\cal C}_x} \le 1\cdot 
\end{align}
 It can be readily verified that the \ac{PGS} outage probability in \eqref{e2e_proper} is non-convex in ${P_{\rm{r}}}$. However, it is shown in \cite[Appendix 3.A]{Khafagy2016FDRDissertation} that the interior of the function is unimodal, and hence quasi-convex, in ${P_{\rm{r}}}$ following similar footsteps of the proof in Theorem \ref{theorem_unimodality}, i.e., via the Descartes rule of signs applied to the derivative of the outage probability function. Hence, the bisection method can be used to locate the global optimum.
 
The second problem is a minimization of a non-convex  function with simple box constraints. Thus, one may try to solve it numerically by, for example, the gradient projected method or the projected Newton's method without any guarantee to converge to an optimal solution \cite{Bertsekas1999nonlinear}. For exact performance analysis purposes, we find the optimal solution via a fine grid search. Moreover, in Theorem \ref{theorem_unimodality}, we proved that the objective function is either a monotonic or unimodal in each of the optimization variables individually over the interior of the constraint set. 
This property motivates us to use a coordinate descent (alternating optimization) method based on a two-dimensional bisection algorithm as in \cite{fadel2012qos}. Fortunately, as it will be noticed in the numerical results section, it always converges numerically to the optimal solution obtained by grid search.        
 
\section{Numerical Results}\label{sec:results}
We numerically evaluate the benefits that can be reaped by employing \ac{IGS} in \ac{FDR} systems. Throughout the following, we compare the performance of \ac{IGS} to that of \ac{PGS} as a benchmark in terms of the outage probability and ergodic rate metrics. For the adaptive design of the \ac{FDR} system based on the outage probability over Rayleigh fading, we use the algorithms discussed in Section \ref{sec: opt}. Specifically, for \ac{PGS}, we show the unoptimized performance with maximum power allocation (MPA), alongside that with optimized relay power using the bisection algorithm (BA) in addition to a fine grid search (GS) for verification purposes. On the other hand, the \ac{IGS} outage performance is shown via two expressions, namely, a) the derived upper bound (UB) in Theorem \ref{theorem_e2e_outage} and, b) the exact end-to-end expression involving the numerical computation of the integral in \eqref{p_out_sr_expectation_rayleigh}. The \ac{IGS} optimization involves two variables; $P_{\rm{r}}$ and $\mathcal{C}_x$. Hence, we consider two cases for \ac{IGS} in the presented figures, namely, i) one-dimensional (1D) optimization over $\mathcal{C}_x$ while adopting maximum power allocation for $P_{\rm{r}}$, and ii) joint $P_{\rm{r}}$/$\mathcal{C}_x$ two-dimensional (2D) optimization. The optimization is done for the two aforementioned cases using both BA and GS, with the prefixes 1D and 2D to distinguish between them.

Also, for performance analysis purposes, we optimize the outage probability and the ergodic rate performances of the \ac{FDR} system over Nakagami-$m$ fading  based on a fine 1D and 2D grid search (GS). We use the simulation parameters in Table I unless otherwise stated. 

\bgroup
\def\arraystretch{1}
\begin{table*}[t!]
\centering
{\rowcolors{1}{black!70!gray!40}{black!70!yellow!40}
\begin{tabular}{ ?c?c?c?c? }
\Xhline{0.6\arrayrulewidth}
Parameter& Value&Parameter& Value   \\
\Xhline{0.6\arrayrulewidth}
$\pi_{\rm{sr}}=\pi_{\rm{rd}}=\pi$ & $20\;\rm{dB}$&$\pi_{\rm{rr}}$&$10\;\rm{dB}$ \\
$\pi_{\rm{sd}}$&$3\;\rm{dB}$&$P_{\rm{s}}=P_{\rm{r}}=P_{\rm{max}}$&$1\;W$    \\
$\mathcal{C}_x$ &$0.9$&$r$&$1\;\rm{bits/sec/Hz}$ \\
$m_{\rm{rr}}=m_{\rm{sd}}$&$1$&$m_{\rm{sr}}=m_{\rm{rd}}$&$m\in\{1,2,3\}$  \\
\Xhline{0.6\arrayrulewidth}
\end{tabular}
}
\caption{Simulation Parameters.}
\label{table:1}
\end{table*}
\egroup
\subsection{Outage Probability Performance}  
Here, we evaluate the proposed lower bound of the end-to-end outage probability for $m \ge 1$ and the upper bound for Rayleigh fading. We use these bounds to optimize the outage performance. Finally, we compare the throughput of the \ac{IGS}-\ac{FDR} to that of \ac{PGS}-\ac{FDR} and \ac{HDR}.  
\subsubsection{Performance of Proposed  Outage Bounds}
\begin{figure*}[!t]
    \centering
    \begin{subfigure}[t]{0.5\textwidth}
        \includegraphics[width=\textwidth]{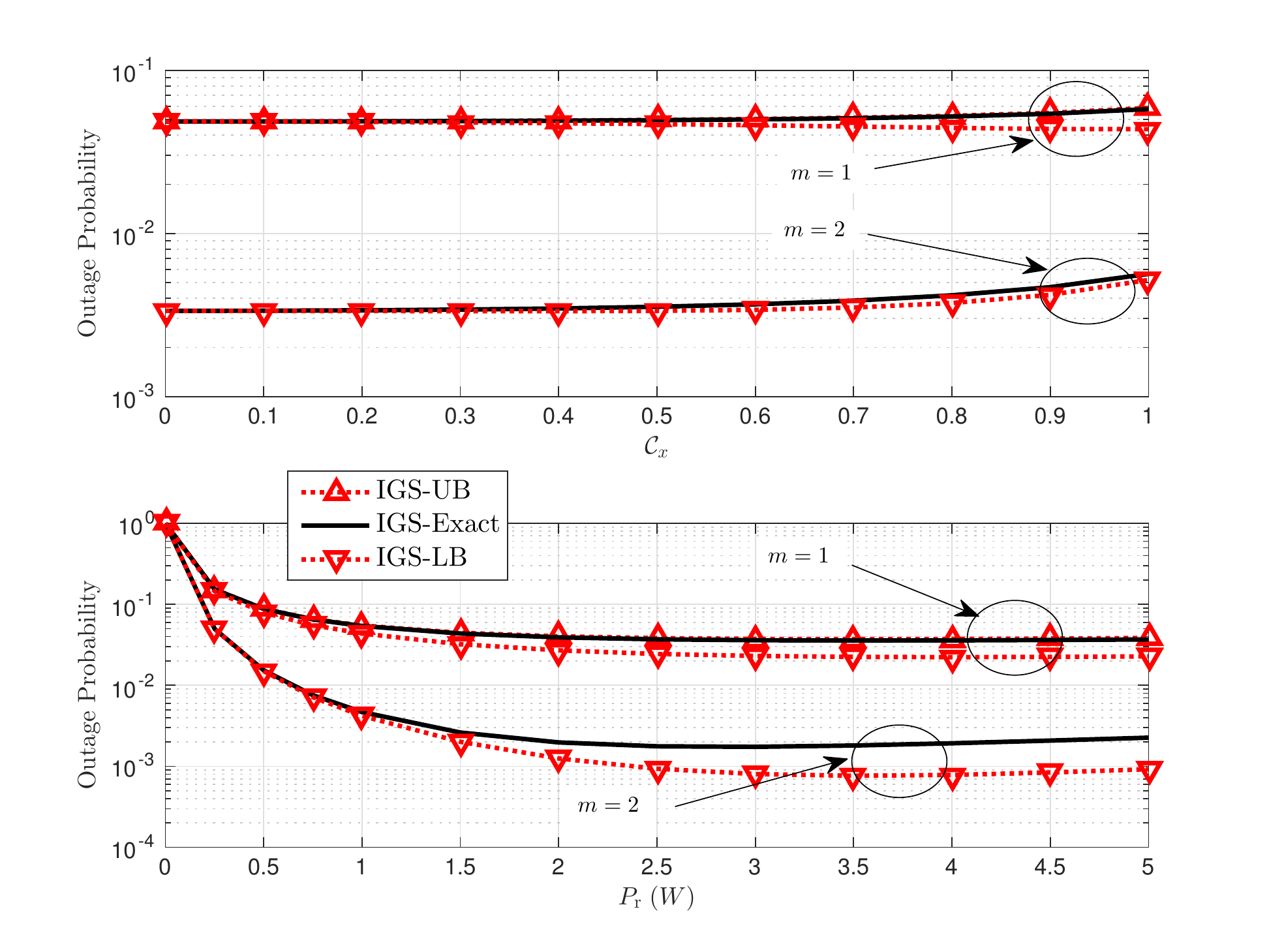}
        \caption{$\pi_{\rm{rr}}=0\;\rm{dB}$}
        \label{OP_vs_Pr_Cx-a}
    \end{subfigure}
    ~ 
 \hfill
    \begin{subfigure}[t]{0.5\textwidth}
        \includegraphics[width=\textwidth]{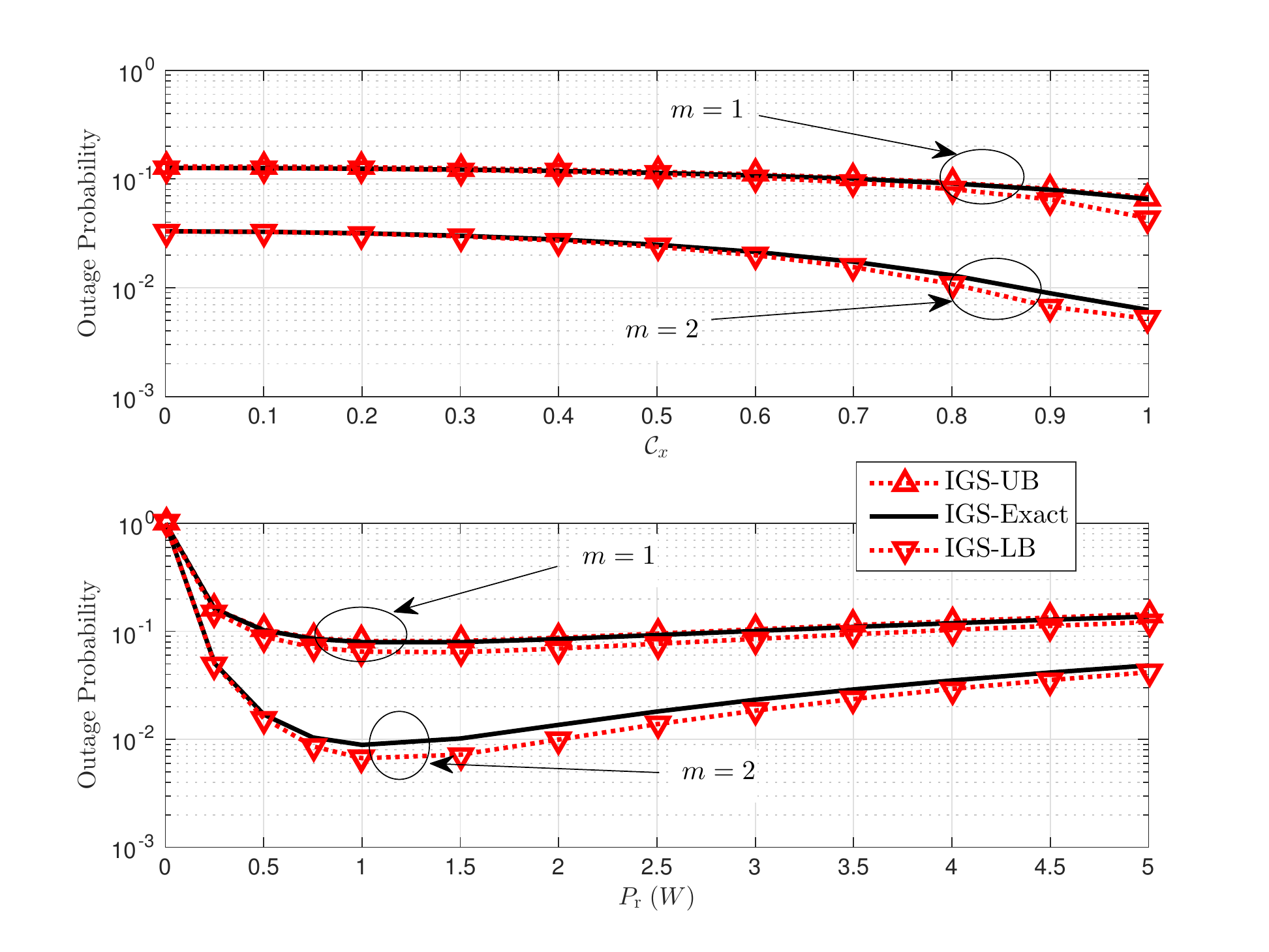}
          \caption{$\pi_{\rm{rr}}=10\;\rm{dB}$}
          \label{OP_vs_Pr_Cx-b}
    \end{subfigure}
    \caption{Outage prob. vs. $P_{\rm{r}}$ [lower] and $\mathcal{C}_x$ [upper], for ~different  $\pi_{\rm{rr}}$ and $m$.}
    \label{OP_vs_Pr_Cx}
\end{figure*}
We evaluate the performance of the proposed outage probability bounds comparing to the exact integral form expressions which are computed numerically. For this purpose, we plot the end-to-end outage probability versus  $P_{\rm{r}}$ and $\mathcal{C}_x$  in Fig. \ref{OP_vs_Pr_Cx} for (a) $\pi_{\rm{rr}}=0\;\rm{dB}$ and (b) $\pi_{\rm{rr}}=10\;\rm{dB}$.  It is clear from the figure that the outage performance improves by increasing $m$ as the fading effect becomes less severe in the relayed path. Also, we can see the tightness of the outage probability bounds for different values of  $P_{\rm{r}}$  and $\mathcal{C}_x$ for the two $\pi_{\rm{rr}}$ values. For a fixed $P_{\rm{r}}$ while increasing $\mathcal{C}_x$, the outage performance depends on the \ac{RSI}. Specifically, in our simulation setup, it is increasing  for $\pi_{\rm{rr}}=0\;\rm{dB}$ and decreasing for $\pi_{\rm{rr}}=10\;\rm{dB}$. Further, for a fixed $\mathcal{C}_x$ while increasing $P_{\rm{r}}$, in both \ac{RSI} values, the outage improves till a certain relay power, at which it starts to deteriorate.



\subsubsection{Outage Probability Optimization}
We evaluate here the performance of the optimized end-to-end outage probability with respect to several system and fading parameters based on the aforementioned optimization algorithms.

\textit{Effect of \ac{RSI}:}   
In Fig. \ref{OP_m_2_opt}, we plot the outage probability versus $\pi_{\rm{rr}}$ for $m=2$ based on a fine 1D and 2D GS. As shown, one can observe that at lower values of the \ac{RSI}, the \ac{IGS} solution reduces to \ac{PGS} since the \ac{RSI} is low and the relay can use more power without deteriorating the $\rm{S-R}$ link quality-of-service. As $\pi_{\rm{rr}}$ increases, the outage performance of the \ac{PGS} is significantly deteriorated. On the other hand, the \ac{IGS} design saturates at a fixed level. However, this level depends on the target rate and the $\rm{S-R}$ and $\rm{R-D}$ link conditions which can be clearly noticed from the outage performance at the two values of ~$\pi$. Moreover, a very interesting observation in Fig. \ref{OP_m_2_opt} is that the individual optimization of the signal asymmetry gives nearly the same performance as that via the joint optimization of $P_{\rm{r}}$  and $\mathcal{C}_x$. Hence, one can resort to the simpler 1D optimization on $\mathcal{C}_x$ to reduce the complexity of the adaptive \ac{FDR} system design.

\begin{figure*}[!t]
\begin{minipage}[t]{0.49\linewidth}
 \includegraphics[width=1\textwidth]{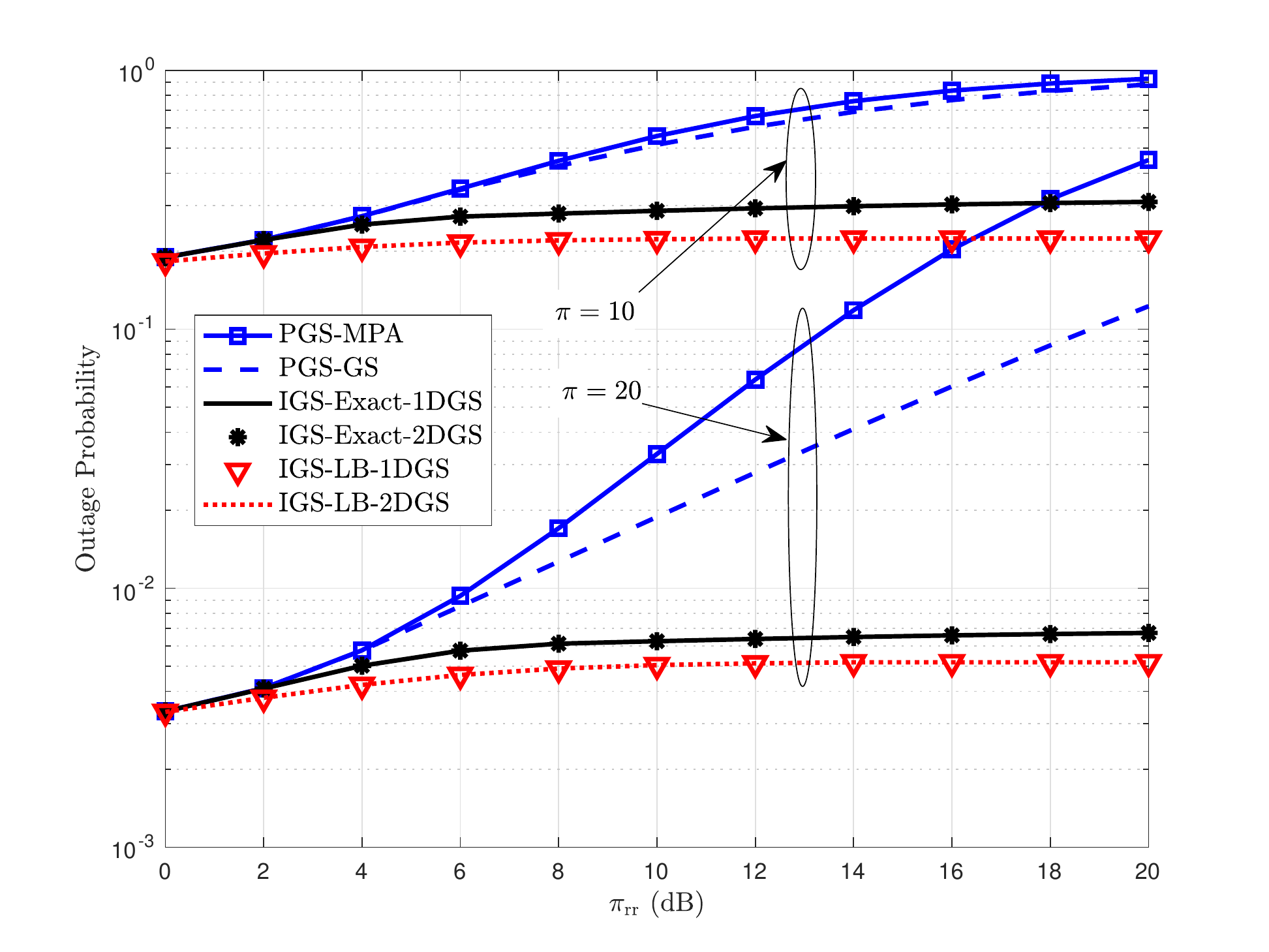}
  \caption{Optimized outage prob.  vs. $\pi_{\rm{rr}}$ for different $\pi$ and  $m=~2$.}
  \label{OP_m_2_opt}
\end{minipage}
\hspace{\fill}
\begin{minipage}[t]{0.49\linewidth}
  
   \includegraphics[width=1\textwidth]{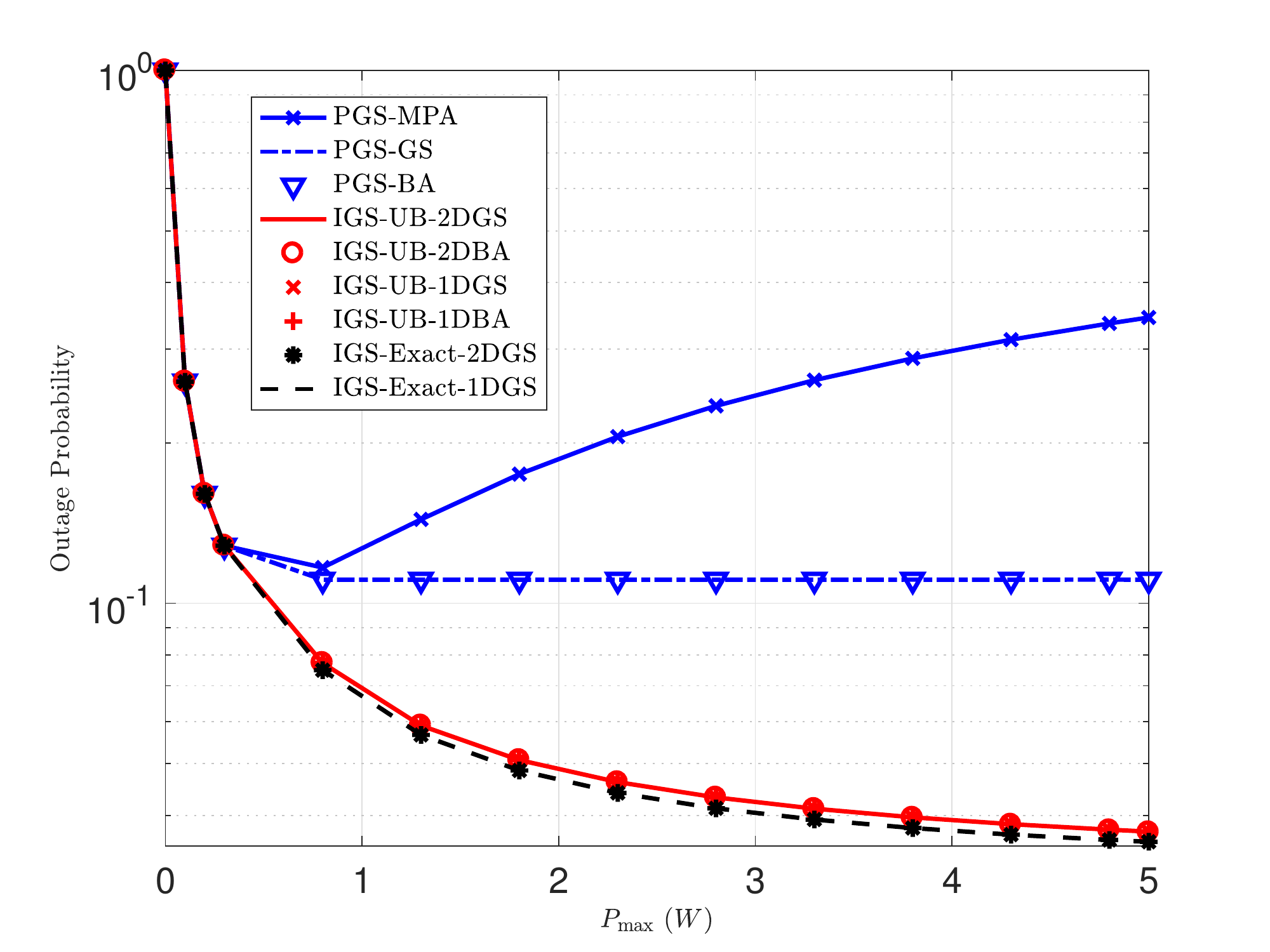}
  \caption{Outage prob. performance vs. relay power budget, for $m=1$.}
  \label{OP_Pmax}
\end{minipage}
\end{figure*}
\textit{Effect of Allowable Relay Power Budget:}    
 In Fig. 4, we study the outage probability  versus the available power budget at the relay. For \ac{FDR} with \ac{PGS}, and specifically when the relay transmits with its maximum power, the outage probability performance is known to be enhanced by increasing the allowable power only till a breakeven point as shown. This point is where the increasing adverse effect of \ac{RSI} on the first hop due to higher relay power starts to exceed any performance returns due to the higher reliability of the second hop. After such a point, any increase in the relay power causes a steady increase in the end-to-end  outage probability. If relay power optimization is allowed in \ac{PGS}, the performance can at best be kept constant after this breakeven point by clipping the transmit power level, rendering any further increase in the power budget unutilized. On the other hand, the performance trend is different when \ac{IGS} is adopted at the relay node. Indeed, by optimizing the relay's circularity coefficient, the outage probability performance continues its decreasing trend. It is also observed that, unlike in \ac{PGS}, the relay tends to use its maximum power in \ac{IGS} when joint power/circularity optimization is considered. For high power budgets, the optimal circularity coefficient value approaches unity, denoting a maximally improper signal that tends to allocate most of its power in only one dimension of the complex signal space. This renders the worst case scenario to have the remaining orthogonal signal space dimension as self-interference-free. The decreasing trend of the outage probability in \ac{IGS}, however, still shows diminishing returns due to the outage performance bottleneck in the first hop, which is primarily influenced by the first hop gain, $\pi_{\rm{sr}}$.

\begin{figure*}[!t]
\begin{minipage}[t]{0.49\linewidth}
  \includegraphics[width=1\textwidth]{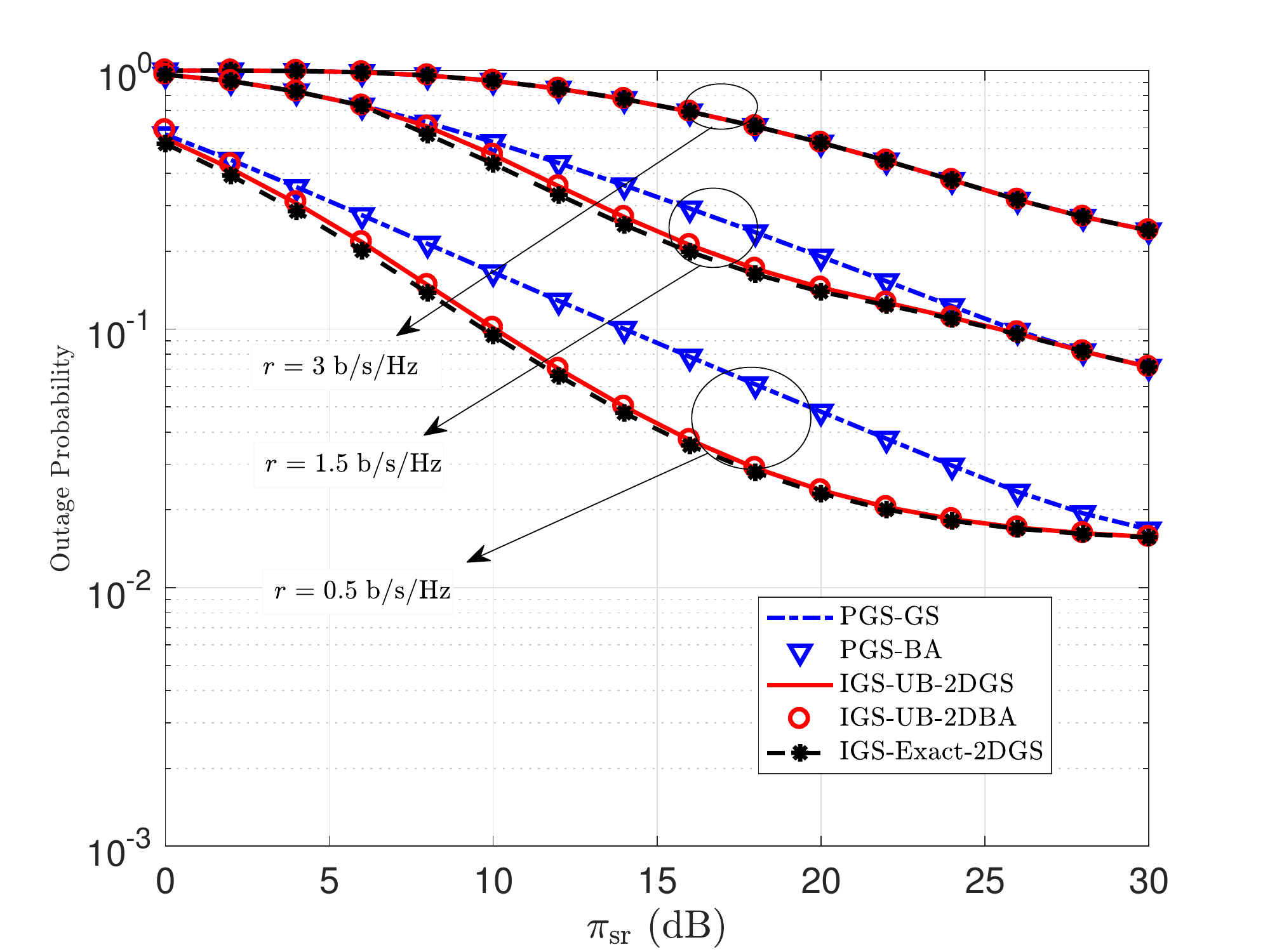}
 \caption{Outage prob. performance vs. $\pi_{\rm{sr}}$ for different target rates and $m=1$.}
  \label{OP_pi_sr}
\end{minipage}
\hspace{\fill}
\begin{minipage}[t]{0.49\linewidth}
  \includegraphics[width=1\textwidth]{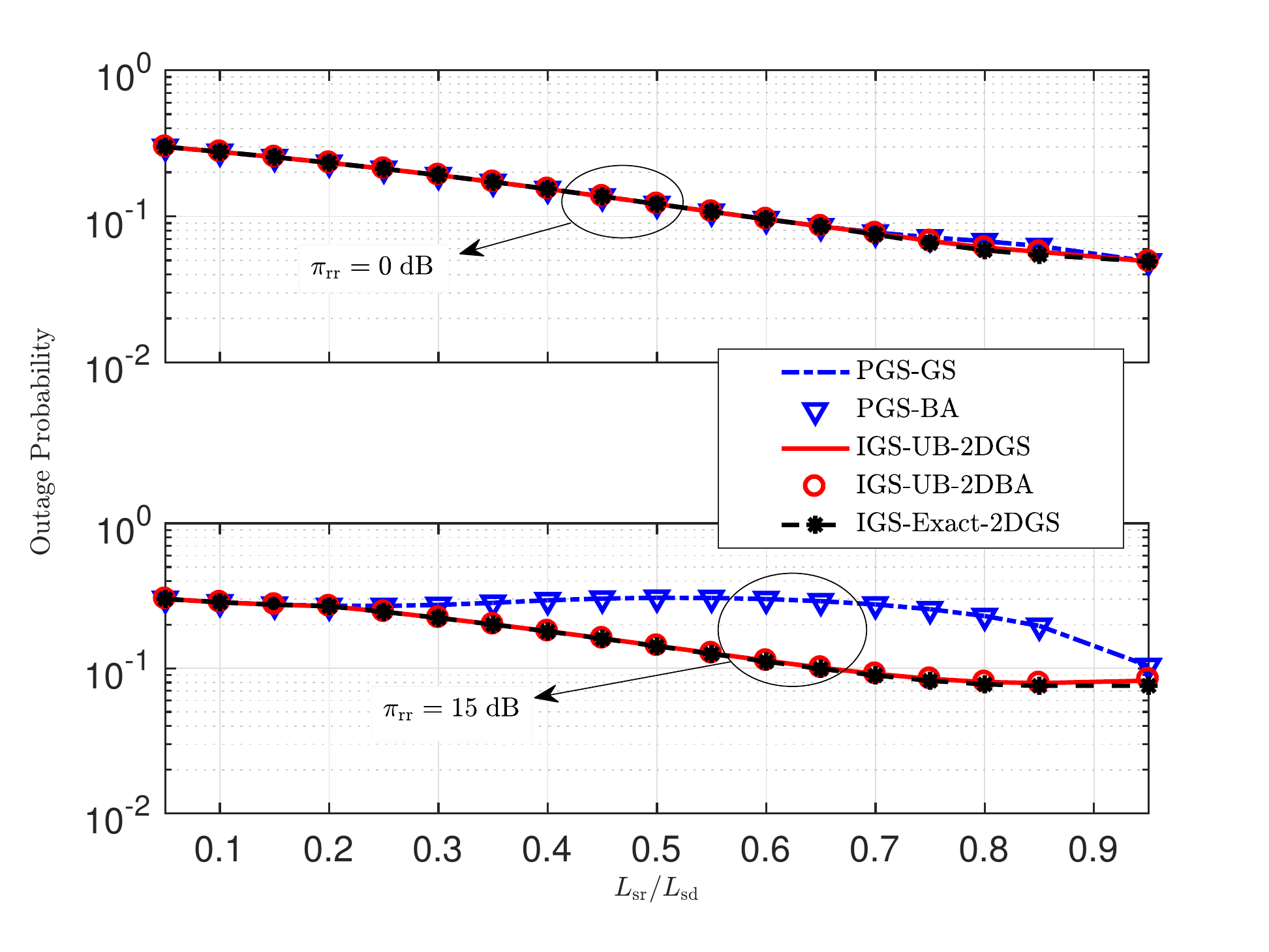}
  \caption{Outage prob. performance vs. normalized relay location, for $m=1$ and $r=0.5\;\rm{bits/sec/Hz}$.}
 \label{OP_loc}
\end{minipage}
\end{figure*}

\textit{Effect of Average $\rm{S-R}$ Link Gain:} In Fig. \ref{OP_pi_sr}, we plot the outage probability versus $\pi_{\rm{sr}}$ for different source target rates. First, communication fails at low $\pi_{\rm{sr}}$ values due to the first hop bottleneck, causing the outage probability of both \ac{PGS} and \ac{IGS} to start close to unity. As $\pi_{\rm{sr}}$ increases, using \ac{IGS} enables the relay to utilize more power relative to \ac{PGS} to boost the outage performance. At the end, when $\pi_{\rm{sr}}$ reaches a significantly higher value than the \ac{RSI}, the first hop no longer operates in the interference-limited regime, and hence, the \ac{IGS} merits become less significant relative to \ac{PGS}. Finally, as shown, the merits of \ac{IGS} over \ac{PGS} are more clear as the target rate decreases. In this case, the rate requirements in the first hop become less stringent, allowing \ac{IGS} to utilize higher transmit power relative to \ac{PGS} and yielding a better performance.   

\textit{Effect of Relay Location:}    
We study the relative relay location impact on the end-to-end outage performance for $\pi_{\mathrm{rr}} \in \{0, \; 15\} \; \mathrm{dB}$. The relay location in Fig. \ref{OP_loc} is defined as the normalized distance of $\rm{S-R}$ link, $L_{\rm{sr}}$ with respect to the distance of $\rm{S-D}$ link, $L_{\rm{sd}}$. When the relay location is closer to the source, the $\rm{S-R}$ link gain is very strong relative to the \ac{RSI}. In such a relatively self-interference-free scenario, the \ac{IGS} solution reduces as expected to the \ac{PGS} solution. As the relay moves towards the destination, the relative adverse effect of \ac{RSI} increases, causing the first hop to operate in the interference-limited regime. In such a regime, the benefits of \ac{IGS} start to show up in mitigating the adverse effect of the \ac{RSI} by tuning the signal impropriety. This gives the performance improvement in the second hop, due to the higher $\rm{R-D}$ link gain, a better opportunity to enhance the end-to-end performance. When the relay is too close to the destination, the \ac{RSI} effect significantly decreases due to the very low relay power required for successful communication, yielding similar \ac{IGS}/\ac{PGS} performance. It is clear that the benefits of \ac{IGS} are noticeable only when the \ac{RSI} link effect is non-negligible. When $\pi_{\mathrm{rr}}=0$ dB, i.e., at the noise level, \ac{IGS} yields the \ac{PGS} solution. 
\subsubsection{Throughput Comparison to \ac{HDR}}
It is essential to compare the performance of \ac{FDR} system with \ac{IGS} not only to \ac{PGS} but also to \ac{HDR}. For this purpose, we consider the optimized network throughput versus the target rate $r$ for (a) $m=2$ and (b) $m=1$  in Fig. \ref{throughput}. First, the optimized  throughput can be computed directly from the optimized end-to-end  outage probability, with a target rate $r$, as ${{\cal T}^*_{{\rm{E - E}}}}\left( {{P_{\rm{r}}},{{\cal C}_x}} \right) = r\left( {1 - {{\cal P}^*_{{\rm{E - E}}}}\left( {{P_{\rm{r}}},{{\cal C}_x}} \right)} \right)$. For the outage expressions of the \ac{HDR}, we consider two protocols;  1) simple \ac{MHDF} \ac{HDR}, where the destination distills the desired information only from the relayed path, and  2) \ac{HDR} with \ac{MRC}, where the destination combines the two time-orthogonal copies of the signal via the direct and relayed paths, as given in ~\cite[Table I]{khafagy2015efficient}. It can be seen from Fig.  \ref{throughput} that the target rate support set is divided into three regions where one protocol outperforms the others; 1) \ac{FDR} with \ac{PGS}, 2) \ac{FDR} with \ac{IGS}, and 3) \ac{HDR}. As shown, for very small (as well as for very high) target rates, optimized \ac{IGS} simply yields the \ac{PGS} solution. As the rate requirement gradually increases, \ac{IGS} can perform better than \ac{PGS} and \ac{HDR} by increasing the signal asymmetry as discussed earlier. However, at higher rates, the \ac{RSI} saturates the \ac{FDR} receiver and the performance deteriorates significantly, even with \ac{IGS} reaching the maximum impropriety, i.e., $\mathcal{C}_x=1$. At this point, \ac{HDR} with/without \ac{MRC} starts to offer better throughput than \ac{FDR}.  We can see that \ac{HDR} with \ac{MRC} is slightly better than \ac{HDR} without \ac{MRC} as expected. Moreover, we can observe from the figure that the \ac{FDR} with \ac{PGS} region is wider at $m=2$. This is caused since the fading is less severe in the relayed path than at $m=1$. Hence, the adverse effect of the \ac{RSI} on the first hop relatively decreases, leaving more room for \ac{FDR} with \ac{IGS} to compete with \ac{HDR} via circularity coefficient tuning. Thus, for lower values of $m$, the performance of \ac{HDR} starts to crossover that of \ac{FDR} at an earlier target rate cutoff than for higher $m$. 
\begin{figure*}[!t]
    \centering
    \begin{subfigure}[b]{0.5\textwidth}
        \includegraphics[width=\textwidth]{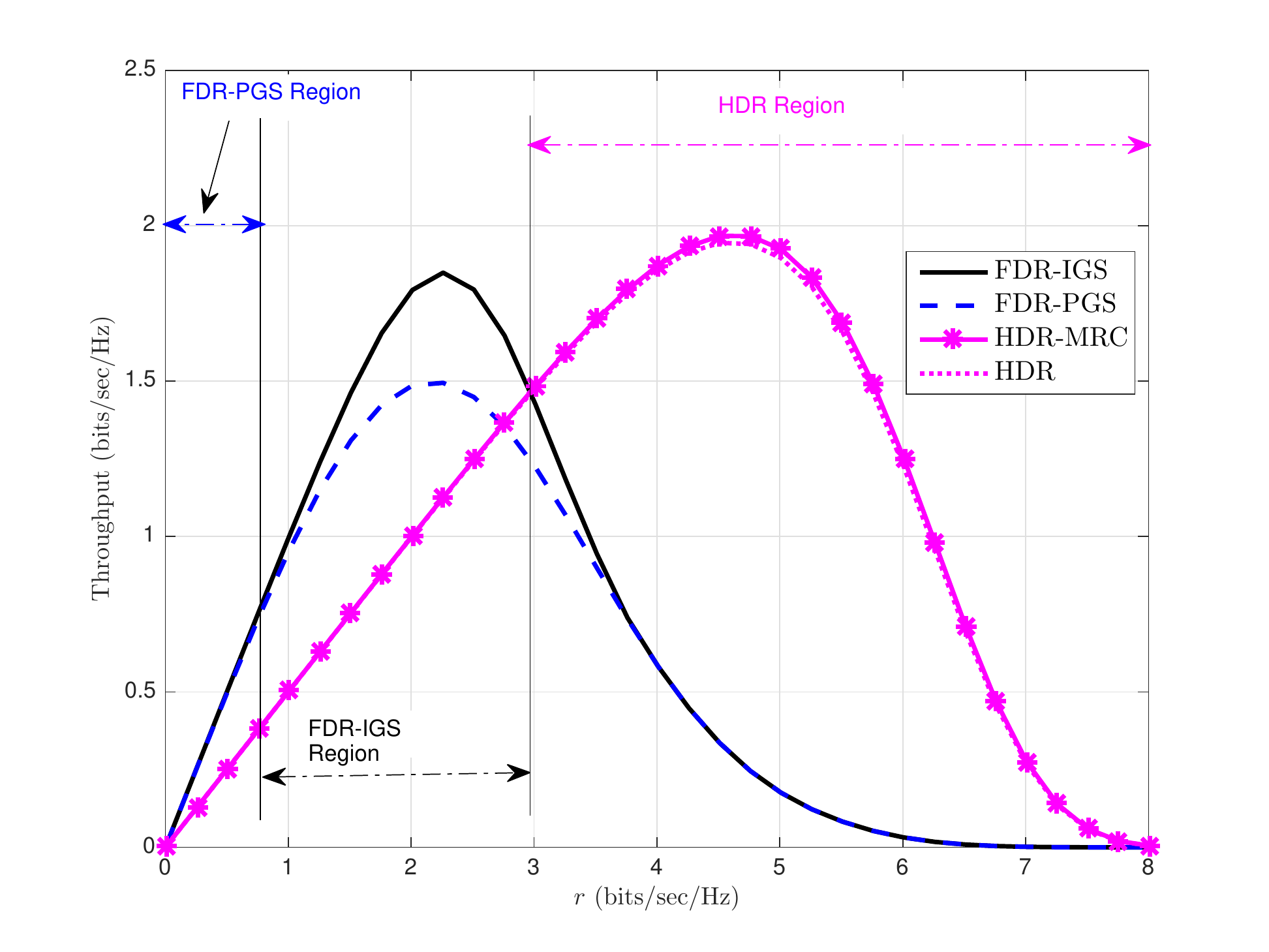}
        \caption{$m=2$}
    \end{subfigure}
    ~ 
      \hfill
    \begin{subfigure}[b]{0.5\textwidth}
        \includegraphics[width=\textwidth]{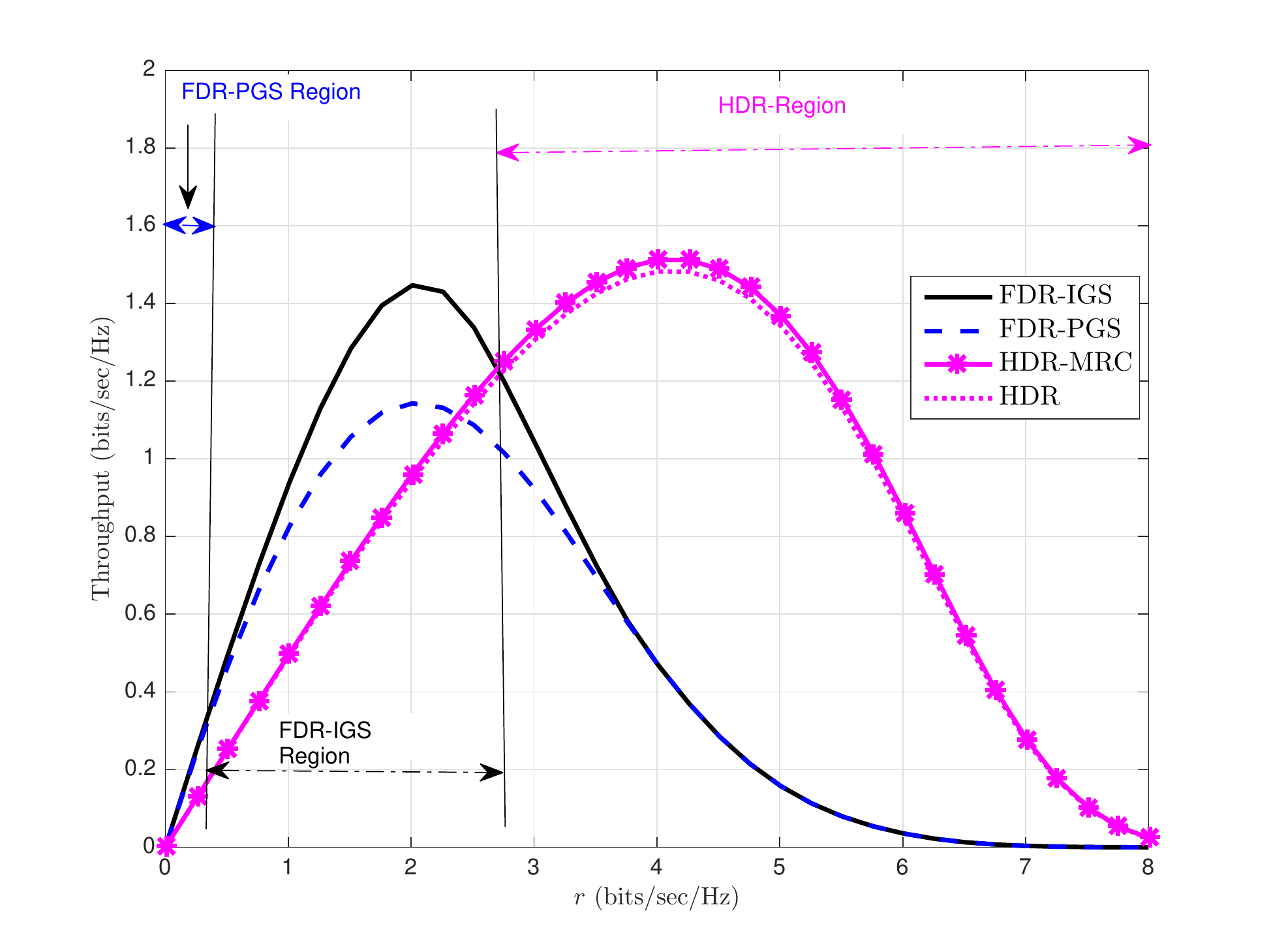}
         \caption{$m=1$ }
    \end{subfigure}
    \caption{Optimized throughput vs. $r$, for $\pi_{\rm{rr}}=15\;\rm{dB}$.}
    \label{throughput}
\end{figure*}

\subsection{Ergodic Rate Performance} 
We first evaluate the performance of the proposed upper bound of the end-to-end ergodic rate for $m \ge 1$ and the lower bound for Rayleigh fading. Then, we analyze the optimized end-to-end ergodic rate based on a fine 1D and 2D GS.
\subsubsection{Performance of Proposed Ergodic Rate Bounds}
We plot the ergodic rate versus $P_{\rm{r}}$ and $\mathcal{C}_x$ for (a) low RSI: $\pi_{\rm{rr}}=0\;\rm{dB}$ and (b) high RSI:  $\pi_{\rm{rr}}=20\;\rm{dB}$ for different $m$ in Fig. \ref{ER_Pr_C_x}. From the figure, one can notice the tightness of the lower and upper bounds compared to the exact ergodic rate. However, the lower bound over Rayleigh fading becomes loose at very high values of $\mathcal{C}_x$, i.e., $\mathcal{C}_x \approx 1$. As $m$ increases, higher end-to-end rate can be achieved as expected. Moreover, similar to Fig. \ref{OP_vs_Pr_Cx}, for a fixed relay transmit power, the value of \ac{RSI} affects the ergodic rate performance. For instance, it is decreasing and increasing for (a) and (b), respectively.

\begin{figure*}[!t]
    \centering
    \begin{subfigure}[b]{0.5\textwidth}
        \includegraphics[width=\textwidth]{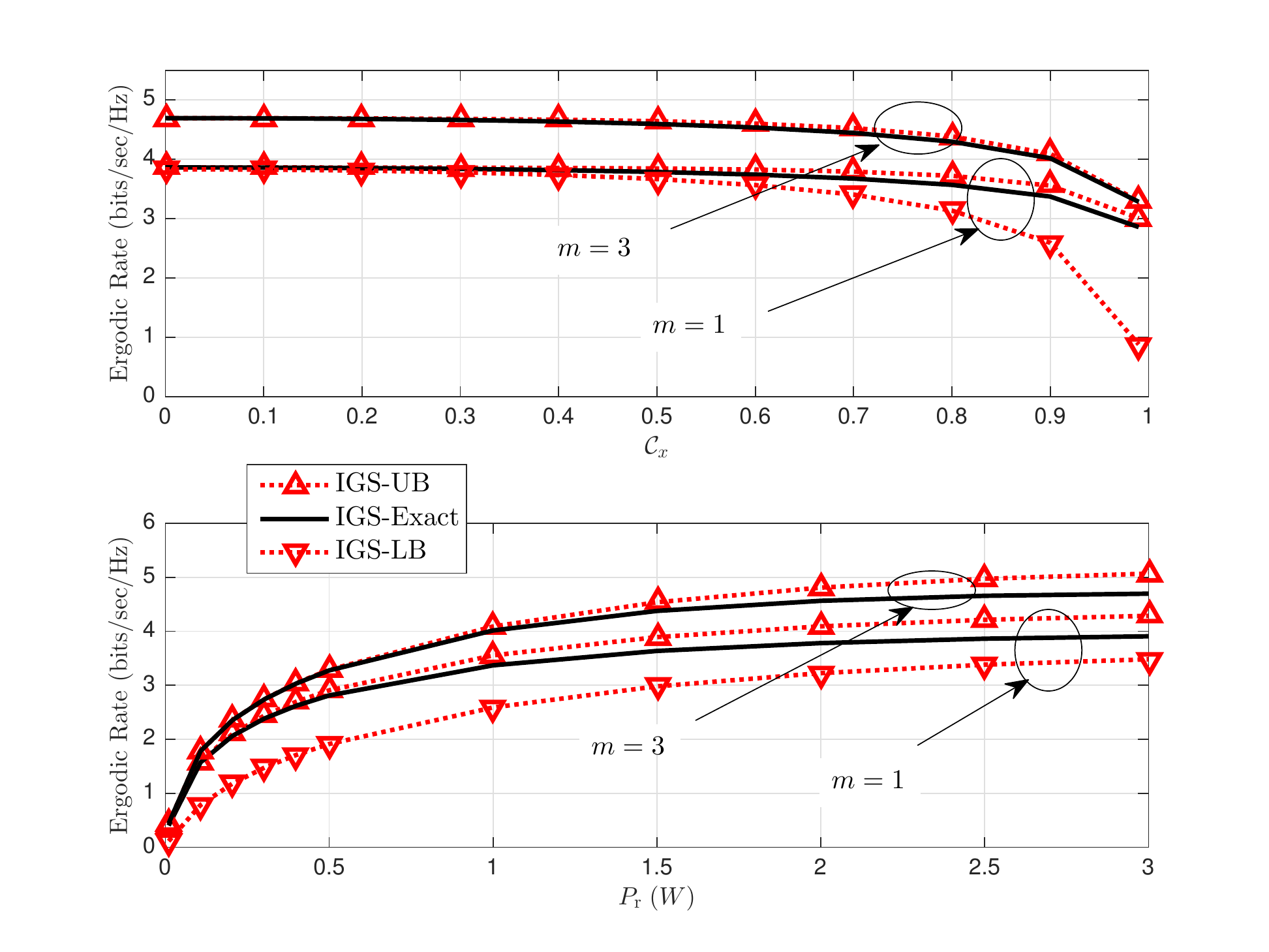}
        \caption{$\pi_{\rm{rr}}=0\;\rm{dB}$}
    \end{subfigure}
    ~ 
      \hfill
    \begin{subfigure}[b]{0.5\textwidth}
        \includegraphics[width=\textwidth]{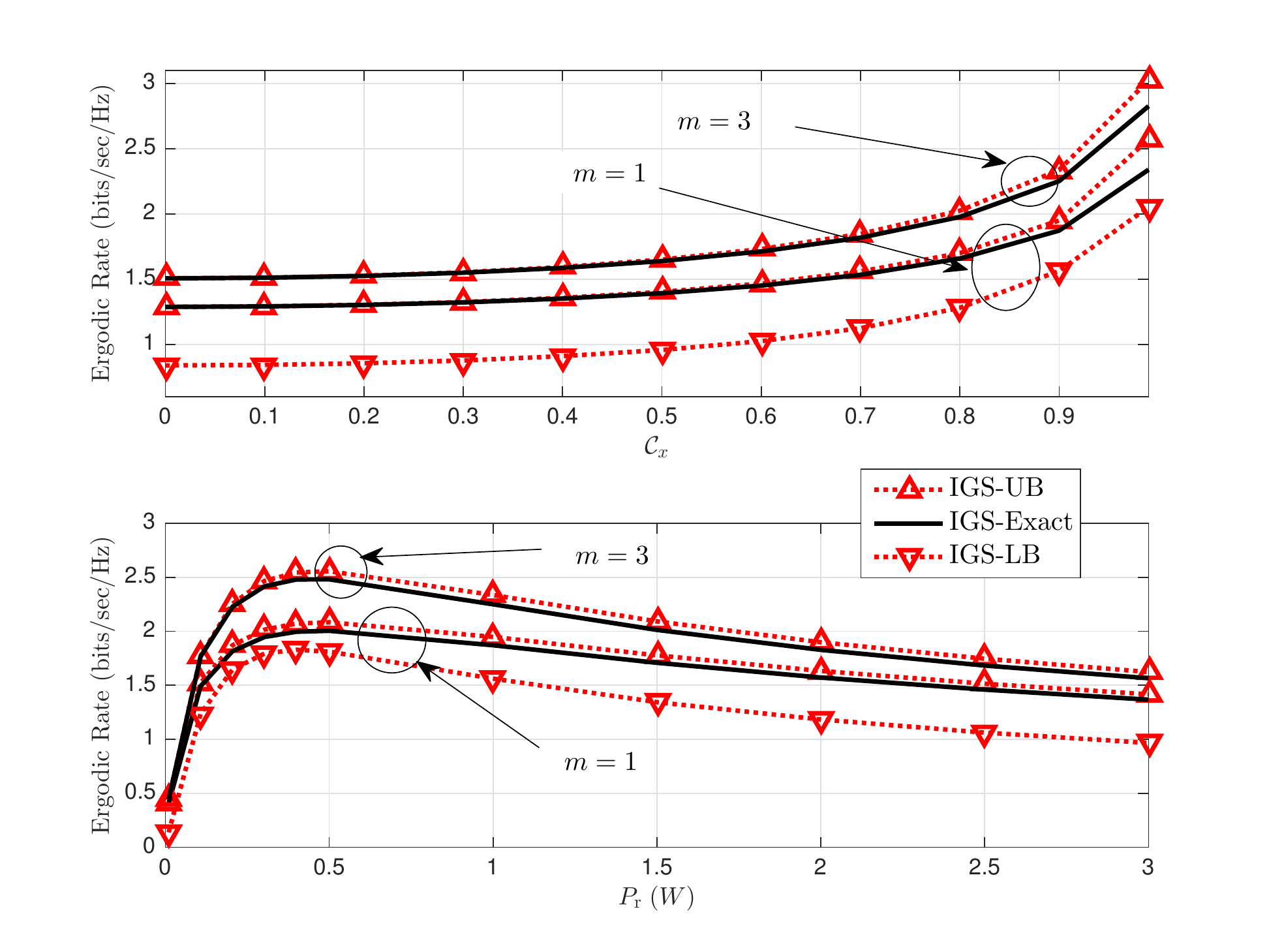}
          \caption{$\pi_{\rm{rr}}=20\;\rm{dB}$}         
    \end{subfigure}
    \caption{Erg. rate vs. $P_{\rm{r}}$ [lower] and $\mathcal{C}_x$ [upper], for ~different  $\pi_{\rm{rr}}$ and $m$.}
     \label{ER_Pr_C_x} 
\end{figure*}

\subsubsection{Ergodic Rate Optimization}
Here, we aim at maximizing the ergodic rate under the box constrains of $P_{\rm{r}}$ and $\mathcal{C}_x$ based on a fine 1D and 2D GS. We plot the  end-to-end ergodic rate versus $\pi_{\rm{rr}}$ for (a) $m=2$ and (b) $m=1$ in Fig. \ref{ER_opt}. It can be noticed that \ac{IGS} reduces to \ac{PGS} at lower \ac{RSI}, similar to the previously noticed outage probability trend in Fig. \ref{OP_m_2_opt}. As the \ac{RSI} increases, the \ac{PGS} performance degrades, while the \ac{IGS} attains nearly a fixed value. Although the gap between the analytical bounds and the exact numerical value widens in (b) to within $0.5$ bits/sec/Hz as the \ac{RSI} increases, the bounds still exhibit the same trend as the exact solution, and hence, they are able to reflect good design insights. For instance, a good figure of the \ac{RSI} level at which \ac{IGS} starts to offer better rates than \ac{PGS} can be obtained to be within the two points via the bounds. Finally, similar to Fig. \ref{OP_m_2_opt}, the optimizing only the circularity coefficient performs nearly the same as the joint optimization.

\begin{figure*}
    \centering
    \begin{subfigure}[b]{0.5\textwidth}
        \includegraphics[width=\textwidth]{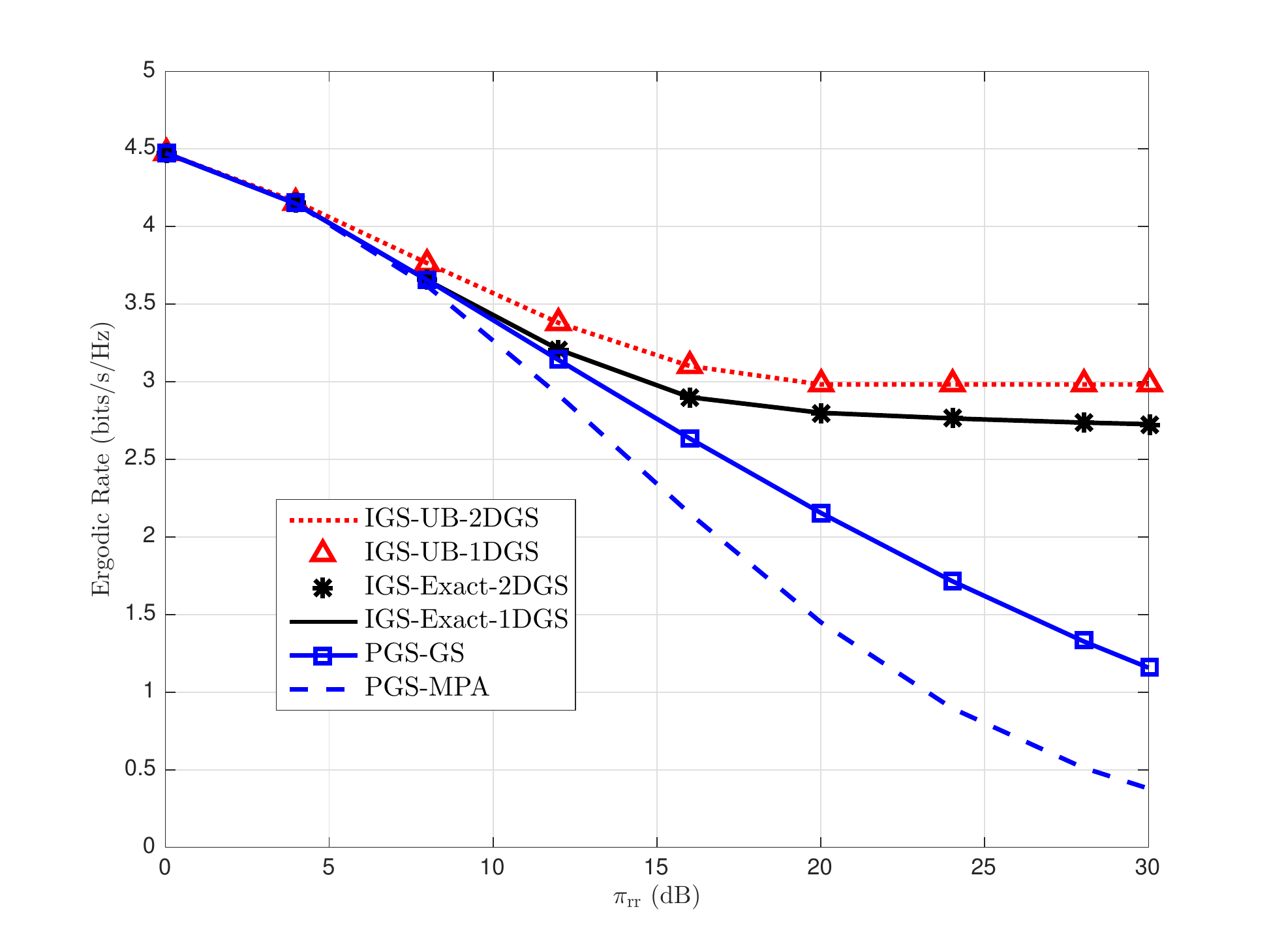}
        \caption{$m=2$}
    \end{subfigure}
    ~ 
      \hfill
    \begin{subfigure}[b]{0.5\textwidth}
        \includegraphics[width=\textwidth]{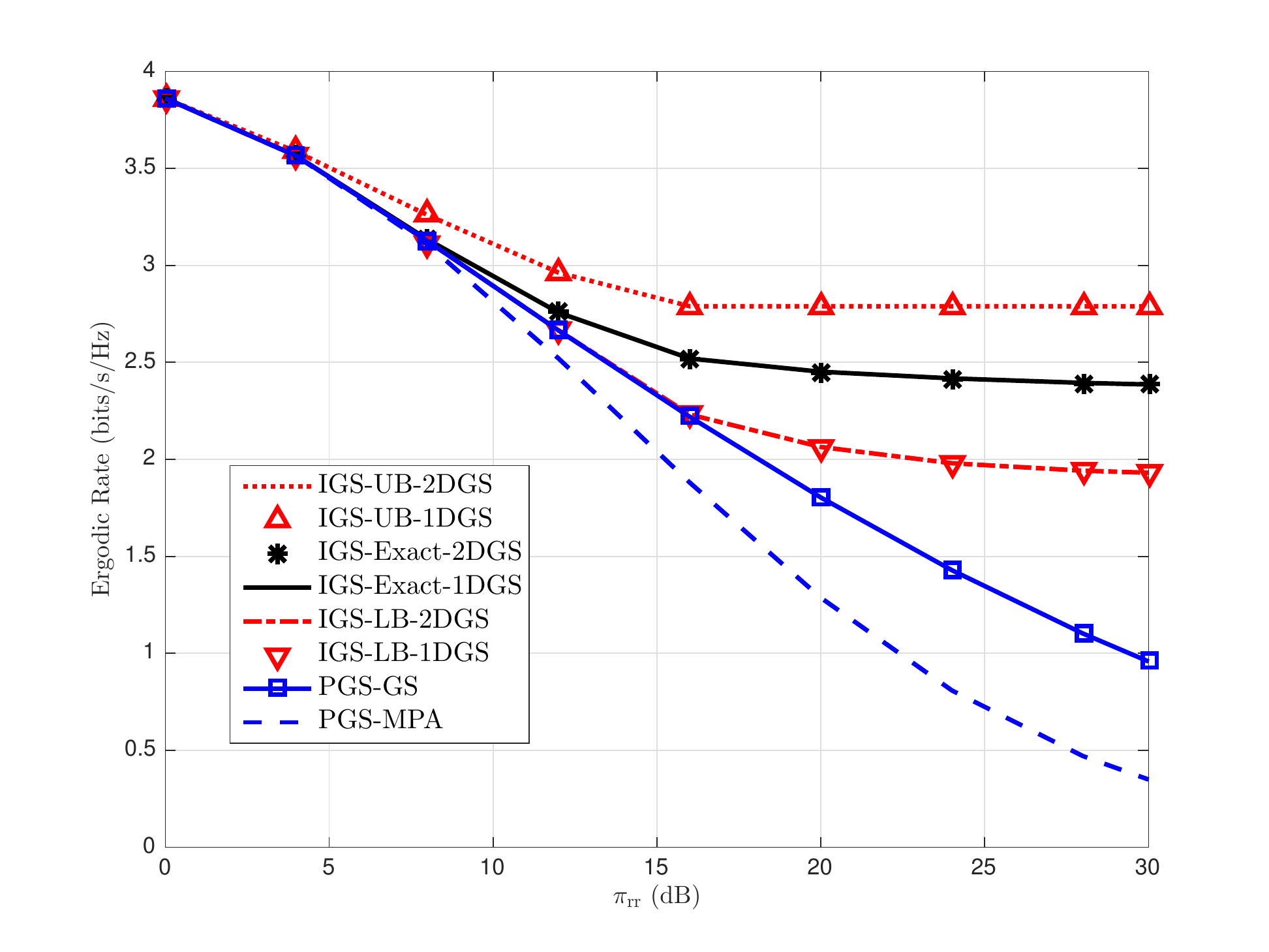}
         \caption{$m=1$ }
    \end{subfigure}
    \caption{Optimized ergodic rate vs. $\pi_{\rm{rr}}$, for different $m$. }
    \label{ER_opt}
\end{figure*}
\glsresetall
\section{Conclusion}\label{sec:conc}
In this work, we study the potential merits of employing \ac{IGS} in \ac{FDR} with non-negligible \ac{RSI} over Nakagami-$m$ fading channels. To analyze the benefits of \ac{IGS}, we present exact integral forms as well as analytical lower and upper bound expressions for the end-to-end outage probability and ergodic rate in terms of the relay's transmit power and circularity coefficient, which measures the degree of the signal impropriety. In order to 
optimize the FDR system performance, we numerically tune the relay  power and circularity coefficient based only on the relay's knowledge of the channel statistics. The findings of this work show that \ac{IGS} yields promising performance merits over \ac{PGS} in \ac{FDR} channels. Specifically, for strong \ac{RSI}, \ac{IGS} tends to leverage higher powers to enhance the performance while alleviating the \ac{RSI} impact by tuning the relay's circularity coefficient. Further, we show that unlike \ac{PGS}, as the \ac{RSI} value increases, \ac{IGS} is able to maintain a fixed performance that depends on the channel statistics and the target rate. It is also shown that, from an end-to-end throughput standpoint, \ac{IGS}-\ac{FDR} can outperform not only \ac{PGS}-\ac{FDR} but also half-duplex relaying with/without maximum ratio combining over certain regions of the target rate. Furthermore, numerical simulations show that it is usually sufficient to individually optimize the circularity coefficient to obtain nearly the same performance as the joint optimization of the power and circularity. Future research lines include obtaining efficient optimization algorithms for the FDR system performance in the Nakagami-$m$ case based on the derived bounds. Moreover,  by capitalizing on the results reported herein, a more general multi-antenna setting with RSI could be considered while adopting IGS at both the source and relay. In this case, it would be interesting, yet challenging, to investigate the IGS benefits via the efficient design of the covariance/pseudo-covariance matrices of the transmit signals.
%
\section*{Appendix A: Proof of Proposition \ref{prop_1}}\label{prop_1_proof}
We prove the convexity of the exponential term inside the expectation operator in \eqref{p_out_sr_expectation_rayleigh} by expressing it as ${{e^{ - f\left( {{g_{\rm{rr}}}} \right)}}}$. In fact, one can show easily that $f\left( {{g_{\rm{rr}}}} \right)$ can be written as
\begin{equation}
f\left( {{g_{\rm{rr}}}} \right) = \sqrt {Ag_{\rm{rr}}^2 + B{g_{\rm{rr}}} + C}  - (D{g_{\rm{rr}}} + F),
\end{equation}
where $A=\frac{{P_{\rm{r}}}^2 \left(1+\gamma (1-{{\cal{C}}_x}^2)\right)}{{P_{\rm{s}}}^2 {\pi_{\rm{sr}}}^2}$, $B=\frac{2 (1+\gamma) {P_{\rm{r}}}}{{P_{\rm{s}}}^2 {\pi_{\rm{sr}}}^2}$, $C=\frac{(1+\gamma)}{{P_{\rm{s}}}^2 {\pi_{\rm{sr}}}^2}$, $D=\frac{{P_{\rm{r}}}}{{P_{\rm{s}}} {\pi_{\rm{sr}}}}$, and $F=\frac{1}{{P_{\rm{s}}} {\pi_{\rm{sr}}}}$ are positive. Indeed, the second derivative of $f\left( {{g_{\rm{rr}}}} \right)$ with respect to $g_{\rm{rr}}$ is
\begin{equation}
\frac{{{\partial ^2}f\left( {{g_{\rm{rr}}}} \right)}}{{\partial g_{\rm{rr}}^2}} = \frac{{4AC - {B^2}}}{{4{{\left( {C + {g_{\rm{rr}}}\left( {B + A{g_{\rm{rr}}}} \right)} \right)}^{3/2}}}} \leq 0,
\end{equation}  
since $1+\gamma(1-{{\cal{C}}_x}^2)\leq 1+\gamma$ for ${0\leq {\cal{C}}_x}\leq 1$. Hence, $f\left( {{g_{\rm{rr}}}} \right)$ is concave and ${{e^{ - f\left( {{g_{\rm{rr}}}} \right)}}}$ is convex.
\vspace{-6pt}
\section*{Appendix B: Proof of Theorem \ref{theorem_unimodality}}\label{theorem_unimodality_proof}
\subsubsection{Unimodality in ${{\cal{C}}_x}$}
The outage upper bound in \eqref{p_out_UB} as a function of ${{\cal{C}}_x}$, denoted here as $x$, is given on the form.
\begin{align}
f(x) = 1 - \frac{{{e^{ - a \frac{{\Psi_r \left( x \right)}}{{\left( {1 - x^2} \right)}} - b \Psi_r \left( \alpha x \right)}} }}{{c \frac{{{\Psi_r}\left( x \right)}}{{\left( {1 - x^2} \right)}} + 1}},
\end{align}
where $0\leq x \leq 1$, $a=\frac{1}{P_{\rm{r}} \pi _{\rm{rd}}}$, $b=\frac{{{{P_{\rm{r}}}{\pi _{\rm{rr}}} + 1} }}{{{P_{\rm{s}}}{\pi _{\rm{sr}}}}}$ and $c=\frac{P_{\rm{s}} \pi _{\rm{sd}}}{P_{\rm{r}} \pi _{\rm{rd}}}$. In the following, we analyze the stationary points of $\overline{f}(x)=1-f(x)$. Its derivative is given by
\begin{align}
\frac{d \overline{f}(x)}{d x} = x \frac{ e^{-a \frac{\Psi_r\left(x\right)}{1-x^2}-b \Psi_r \left( \alpha x \right)}}{\left(c\frac{\Psi_r\left(x\right)}{1-x^2}+1\right)^2} S(x),
\end{align}
\vspace{-5pt}
where
\begin{align}
S(x) =&\left(c\frac{\Psi_r\left(x\right)}{1-x^2}\!+1\!\right)\!\!\Bigg[\!\frac{a \left(2 \Psi_r\left(x\right)+\gamma 
   \left(x^2-1\right)\right)}{(\Psi_r\left(x\right)+1) \left(1-x^2\right)^2}\!\nonumber \\
   &+\!\frac{b \gamma
    \alpha^2}{\Psi_r \left( \alpha x \right)+1} \!\Bigg]\!+\!\frac{\gamma  c}{(\Psi_r\left(x\right)+1)
   \left(1-x^2\right)}\!-\!\frac{2 c \Psi_r\left(x\right)}{\left(1-x^2\right)^2}\cdot
   \end{align}
From the given form, and in addition to the roots of $S(x)$, it is clear that $\frac{d \overline{f}(x)}{d x}$ admits only a zero at $x=0$. Now, we investigate the roots for $S(x)$, and use the change of variables, $z={\Psi_r}\left( x \right) + 2$. Hence, $1-x^2 = \frac{z (z-2)}{\gamma}$. After substitution and some manipulations, $S(z)$ is hence given for our region of interest, $2 \leq z \leq 1+ \sqrt{1+\gamma}$, by{
\begin{align}
S(z) =\left(c\frac{\gamma}{z}+1\right) \hspace{-3pt}\left(\frac{- a \gamma^2}{z^2 (z-1)}+\frac{b \gamma
    \alpha^2}{\Psi_r \left( \alpha x \right)+1}\right)-\frac{\gamma^2  c}{z^2 (z-1)}\cdot
   \end{align}
Since $0<\alpha<1$, we know that $1-\alpha^2 x^2\geq 1-x^2$. Hence, $\Psi_r \left(\alpha x\right)+1\geq\Psi_r \left(x\right)+1=z-1$. Let $\Psi_r \left(\alpha x\right)+1=t_z(z-1)$, where $t_z\geq1$. Therefore,
\begin{eqnarray}
S(z) &=& \frac{(c \gamma+z)(- a \gamma^2 t_z + b \gamma
    \alpha^2 z^2)-\gamma^2  c t_z z}{t_z z^3 (z-1)}\cdot
\end{eqnarray}
The numerator is a cubic polynomial in $z$ which is given by
$
T(z) = b \alpha^2 \gamma z^3     + b \alpha^2 c \gamma^2 z^2 - (a+c) \gamma^2 t_z z - a c \gamma^3 t_z\cdot
$
To find the number of positive roots for $T(z)$, we use the Descartes rule of signs \cite{prasolov2009polynomials}. Specifically, for the sequence formed by the descending order of the cubic equation coefficients, i.e., the sequence  $\{b \alpha^2 \gamma, b \alpha^2 c \gamma^2, - (a+c) \gamma^2 t_z, - a c \gamma^3 t_z\}$,  the number of sign changes is only one. For our real cubic polynomial, this determines the number of positive roots to be exactly one root. Hence, in the positive region of interest, $2 \leq z \leq 1 + \sqrt{1 + \gamma}$, either one or no \emph{feasible} roots exist for $T(z)$, and hence for $S(z)$. This shows that $\overline{f}(x)$ is either monotonic or unimodal due to the existence of one root at maximum in its interior.

\vspace{5pt}
\subsubsection{Unimodality in ${P_{\rm{r}}}$}
The outage probability upper bound in \eqref{p_out_UB} as a function of ${P_{\rm{r}}}$, also denoted here as $x$, is given by
\begin{align}
\PUB \left( \PR \right) =  1 - c\underbrace{\frac{e^{ -\Big(\frac{a}{\PR}+ {{b\PR} }\Big)}}{\frac{d}{\PR} + 1}}_{D(\PR)},
\end{align}
where $a=\frac{{\Psi_r \left( {{{\cal C}_x}} \right)}}{{{\pi _{\rm{rd}}}\left( {1 - {\cal C}_x^2} \right)}}$, $b={\pi _{\rm{rr}}}\frac{\Psi_r \left( \alpha \mathcal{C}_x \right)}{{{P_{\rm{s}}}{\pi _{\rm{sr}}}}}$, $c=e^{-\frac{\Psi_r \left( \alpha \mathcal{C}_x \right)}{{{P_{\rm{s}}}{\pi _{\rm{sr}}}}}}$, and $d=\frac{{P_{\rm{s}}}{\pi _{\rm{sd}}}{{\Psi_r}\left( {{{\cal C}_x}} \right)}}{{{\pi _{\rm{rd}}}\left( {1 - {\cal C}_x^2} \right)}}$. 
In what follows, we analyze the stationary points of the interior of the function $D(\PR)$. 
Differentiating $D(\PR)$ with respect to $\PR$ yields
\begin{eqnarray}
\frac{d D}{d \PR} 
&=&\frac{\overbrace{d\PR + (a-b\PR^2)(\PR+d)}^{S(\PR)}}{\PR(\PR+d)^2}e^{-\left(\frac{a}{\PR}+b\PR\right)}.
\end{eqnarray}
Now, the stationary points of $D(\PR)$ originate from the roots of the cubic polynomial $S(x)=-b x^3-bd x^2+(a+d)x+ad$.
Again, there exists only one sign change  in the coefficients of the cubic equation, and hence, only one root exists for $S(\PR)$. Therefore, in the positive region of interest, i.e.,  $0<{P_{\rm{r}}}<{P_{\max }}$, $D(\PR)$ and hence $\PUB (\PR)$, are either monotonic or unimodal in $\PR$.
\bibliographystyle{IEEEtran}

\bibliography{IEEEabrv,mgaafar_ref_October_2015}

\end{document}